\documentclass[sigconf]{acmart}
\usepackage[vlined,boxed]{algorithm2e}
\usepackage{amsmath}
\usepackage{amsfonts}
\usepackage{epsfig}
\usepackage{graphics}
\usepackage{color}
\usepackage{comment}
\usepackage{paralist}
\usepackage{mathtools}
\usepackage{multirow}
\usepackage{subfigure}
\usepackage{eepic}
\usepackage{stmaryrd}
\usepackage{textcomp}
\usepackage{balance}
\usepackage[inline]{enumitem}
\usepackage{ulem}

\begin{document}

\newcommand{\from}[3]{\textbf{FROM} #1 \textbf{TO} #2: #3}
\def\FP{\text{\rm FP}}
\def\spanL{\text{\rm SpanL}}
\def\spanLL{\text{\rm SpanLL}}
\def\sharpP{\text{\rm \#P}}
\def\sharpL{\text{\rm \#L}}
\def\NL{\text{\rm NL}}
\def\PTIME{\text{\rm P}}
\def\PH{\text{\rm PH}}
\def\NP{\text{\rm NP}}
\def\LOGSPACE{\text{\rm L}}
\def\NSPACE{\rm NSPACE}
\def\co{\rm co\text{-}}
\def\PSPACE{\rm PSPACE}
\def\EXPTIME{\rm EXPTIME}
\def\TWOEXPTIME{\rm 2EXPTIME}
\def\AEXSPACE{\rm AEXSPACE}
\def\ACZ{\rm AC_0}
\def\hard{\rm \text{-}hard}
\def\complete{\text{-{\rm complete}}}

\newcommand{\depth}[1]{\mathsf{depth}(#1)}
\newcommand{\rank}[1]{\mathsf{rank}(#1)}
\newcommand{\sdepth}[1]{\star\text{-}\mathsf{depth}(#1)}
\newcommand{\mh}[1]{\mathsf{mh}(#1)}
\newcommand{\ma}[1]{\mathsf{ma}(#1)}
\newcommand{\md}[1]{\mathsf{md}(#1)}
\newcommand{\smd}[1]{\star\text{-}\mathsf{md}(#1)}
\newcommand{\omd}[1]{\obl\text{-}\mathsf{md}(#1)}
\newcommand{\gforest}[1]{\mathsf{gforest}(#1)}
\newcommand{\gtree}[1]{\mathsf{gtree}(#1)}
\newcommand{\slinpath}[1]{simple linear $#1$-path}
\newcommand{\linpath}[1]{linear $#1$-path}
\newcommand{\guardedpath}[1]{guarded $#1$-path}
\newcommand{\stickypath}[1]{sticky $#1$-path}
\newcommand{\R}{\mathcal{R}}
\newcommand{\Lang}{\mathcal{L}}
\newcommand{\mi}[1]{\mathit{#1}}
\newcommand{\ins}[1]{\mathbf{#1}}
\newcommand{\adom}[1]{\mathsf{dom}(#1)}
\newcommand{\ra}{\rightarrow}
\newcommand{\fr}[1]{\mathsf{fr}(#1)}
\newcommand{\dep}{\Sigma}
\newcommand{\sch}[1]{\mathsf{sch}(#1)}
\newcommand{\sign}{\ins{S}}
\newcommand{\body}[1]{\mathsf{body}(#1)}
\newcommand{\head}[1]{\mathsf{head}(#1)}
\newcommand{\guard}[1]{\mathsf{guard}(#1)}
\newcommand{\class}[1]{\mathsf{#1}}
\newcommand{\pos}[1]{\mathsf{pos}(#1)}
\newcommand{\spos}[1]{\mathsf{spos}(#1)}
\newcommand{\app}[2]{\langle #1,#2 \rangle}
\newcommand{\tup}[1]{\langle #1 \rangle}
\newcommand{\crel}[1]{\prec_{#1}}
\newcommand{\tcrel}[1]{\prec_{#1}^{\star}}
\newcommand{\rctaa}{\class{CT}_{\forall \forall}^{\mathsf{res}}}
\newcommand{\rctaapr}{\mathsf{CT}_{\forall \forall}^{\mathsf{res}}}
\newcommand{\rctae}{\class{CT}_{\forall \exists}^{\mathsf{res}}}
\newcommand{\rctaepr}{\mathsf{CT}_{\forall \exists}^{\mathsf{res}}}
\newcommand{\base}[1]{\mathsf{base}(#1)}
\newcommand{\eqt}[1]{\mathsf{eqtype}(#1)}
\newcommand{\var}[1]{\mathsf{var}(#1)}
\newcommand{\const}[1]{\mathsf{const}(#1)}
\newcommand{\result}[2]{\mathsf{result}(#1,#2)}
\newcommand{\soresult}[2]{\sobl\text{-}\mathsf{result}(#1,#2)}
\newcommand{\oresult}[2]{\obl\text{-}\mathsf{result}(#1,#2)}
\newcommand{\sresult}[2]{\star\text{-}\mathsf{result}(#1,#2)}

\newcommand{\reach}[1]{\rightsquigarrow_{#1}}
\newcommand{\obl}{\mathsf{o}}
\newcommand{\sobl}{\mathsf{so}}
\newcommand{\std}{\mathsf{std}}
\newcommand{\cta}[1]{\class{CT}_{\forall \forall}^{#1}}
\newcommand{\cte}[1]{\class{CT}_{\forall \exists}^{#1}}

\newcommand{\ctda}[2]{\class{CT}_{\forall,#2}^{#1}}
\newcommand{\ctde}[2]{\class{CT}_{\exists,#2}^{#1}}

\newcommand{\ct}[1]{\class{CT}_{\forall}^{#1}}
\newcommand{\ctd}[2]{\class{CT}^{#1}_{#2}}
\newcommand{\ctapr}[1]{\mathsf{CT}_{\forall \forall}^{#1}}
\newcommand{\ctepr}[1]{\mathsf{CT}_{\forall \exists}^{#1}}
\newcommand{\ctdapr}[1]{\mathsf{CT}_{\forall}^{#1}}
\newcommand{\ctdepr}[1]{\mathsf{CT}_{\exists}^{#1}}
\newcommand{\ctpr}[1]{\mathsf{CT}_{\forall}^{#1}}
\newcommand{\ctdpr}[1]{\mathsf{CT}^{#1}}
\newcommand{\cri}[1]{\mathsf{cr}(#1)}
\newcommand{\lin}[1]{\mathsf{Lin_{\class{S}}}(#1)}
\newcommand{\ling}[1]{\mathsf{lin}(#1)}
\newcommand{\shape}[1]{\mathsf{shape}(#1)}
\newcommand{\svar}[1]{\mathsf{svar}(#1)}
\newcommand{\constfree}[1]{\mathsf{c\text{-}free}(#1)}
\newcommand{\id}[2]{\mathsf{id}_{#1}(#2)}
\newcommand{\unique}[1]{\mathsf{unique}(#1)}
\newcommand{\simple}[1]{\mathsf{simple}(#1)}
\newcommand{\gsimple}[1]{\mathsf{gsimple}(#1)}
\def\sub{\sqsubseteq}
\def\substrict{\sqsubset}

\newcommand{\norm}[1]{\mathsf{Norm}(#1)}
\newcommand{\depg}[1]{\mathsf{dg}(#1)}
\newcommand{\edepg}[1]{\mathsf{edg}(#1)}
\newcommand{\sodg}[1]{\mathsf{so\text{-}dg}(#1)}
\newcommand{\odg}[1]{\mathsf{o\text{-}dg}(#1)}
\newcommand{\ex}[1]{\mathsf{exvar}(#1)}
\newcommand{\mgu}[2]{\mathsf{mgu}(#1,#2)}
\newcommand{\res}[1]{\mathsf{res}(#1)}
\newcommand{\f}[2]{f_{#1}(#2)}
\newcommand{\arity}[1]{\mathsf{ar}(#1)}
\newcommand{\atoms}[1]{\mathsf{atoms}(#1)}
\newcommand{\birth}[2]{\mathsf{birth}_{#1}(#2)}
\newcommand{\OMIT}[1]{}
\newcommand{\crt}[1]{\texttt{cr}(#1)}
\newcommand{\precd}[1]{\prec_{#1}}
\newcommand{\lvl}[2]{\texttt{lv}_{#1}(#2)}
\newcommand{\posvar}[2]{\mathsf{pos}(#1,#2)}
\newcommand{\posterm}[2]{\mathsf{pos}(#1,#2)}
\newcommand{\varpos}[2]{\mathsf{var}(#1,#2)}
\newcommand{\termpos}[2]{\mathsf{term}(#1,#2)}
\newcommand{\CT}[2]{\mathsf{CT}^{#1}_{#2}}
\newcommand{\frpos}[1]{\mathsf{frpos}(#1)}
\newcommand{\dom}{\mathbf{C}}
\newcommand{\freshdom}{\mathbf{N}}
\newcommand{\frontier}[1]{\mathsf{fr}(#1)}
\newcommand{\eqtype}[1]{\mathsf{eqtype}(#1)}
\def\iso{\simeq}
\newcommand{\can}[1]{\mathsf{can}(#1)}
\newcommand{\proj}[2]{\Pi_{#1}(#2)}
\newcommand{\pred}[1]{\mathit{pred}(#1)}
\newcommand{\predt}[1]{[#1]}
\newcommand{\atom}[1]{\underline{#1}}
\newcommand{\tuple}[1]{\bar{#1}}
\newcommand{\resolv}[1]{[#1]}
\newcommand{\parent}[1]{\mathit{par}(#1)}
\newcommand{\dept}[1]{\mathit{depth}(#1)}
\newcommand{\chase}[2]{\mathsf{chase}(#1,#2)}
\newcommand{\chasesize}[2]{\mathsf{chsize}(#1,#2)}
\newcommand{\starchasei}[3]{\star\text{-}\mathsf{chase}^{#3}(#1,#2)}
\newcommand{\starchase}[2]{\star\text{-}\mathsf{chase}(#1,#2)}
\newcommand{\sochasei}[3]{\sobl\text{-}\mathsf{chase}^{#3}(#1,#2)}
\newcommand{\sochase}[2]{\sobl\text{-}\mathsf{chase}(#1,#2)}
\newcommand{\ochasei}[3]{\obl\text{-}\mathsf{chase}^{#3}(#1,#2)}
\newcommand{\ochase}[2]{\obl\text{-}\mathsf{chase}(#1,#2)}
\newcommand{\completion}[2]{\mathsf{complete}(#1,#2)}
\newcommand{\mar}[1]{\hat{#1}}
\newcommand{\nullobl}[3]{\bot^{#1}_{#2,#3}}
\newcommand{\nullsobl}[3]{\bot^{#1}_{#2,#3_{|\frontier{#2}}}}
\newcommand{\startype}[1]{\star\textrm{-}\mathsf{type}(#1)}
\newcommand{\type}[2]{\mathsf{type}_{#1}(#2)}
\newcommand{\types}[2]{#1\textrm{-}\mathsf{types}(#2)}
\newcommand{\src}[1]{\mathsf{src}(#1)}
\newcommand{\IDTGD}{\class{ID}}
\newcommand{\DLLITETGD}{\mathsf{DL\textrm{-}Lite^{TGD}}}
\newcommand{\SLTGD}{\class{SL}}
\newcommand{\LTGD}{\class{L}}
\newcommand{\GTGD}{\class{G}}
\newcommand{\WGTGD}{\class{WG}}
\newcommand{\RATGD}{\class{RA}}
\newcommand{\LARATGD}{\class{LARA}}
\newcommand{\LCRATGD}{\class{LCRA}}
\newcommand{\LCWATGD}{\class{LCWA}}
\newcommand{\WATGD}{\class{WA}}
\newcommand{\DAT}{\class{DAT}}
\newcommand{\UCQ}{\class{UCQ}}
\newcommand{\SLRATGD}{\class{SLRA}}
\newcommand{\SLWATGD}{\class{SLWA}}
\newcommand{\LRATGDP}{\class{LRA}^{+}}
\newcommand{\LWATGDP}{\class{LWA}^{+}}
\newcommand{\oblrew}[1]{\mathsf{enrichment}(#1)}

\newcommand{\termeq}{\equiv}
\newcommand{\skolem}[3]{\bot_{#1,#2}^{#3}}

\def\qed{\hfill{\qedboxempty}      
  \ifdim\lastskip<\medskipamount \removelastskip\penalty55\medskip\fi}

\def\qedboxempty{\vbox{\hrule\hbox{\vrule\kern3pt
                 \vbox{\kern3pt\kern3pt}\kern3pt\vrule}\hrule}}

\def\qedfull{\hfill{\qedboxfull}   
  \ifdim\lastskip<\medskipamount \removelastskip\penalty55\medskip\fi}

\def\qedboxfull{\vrule height 4pt width 4pt depth 0pt}

\newcommand{\markfull}{\qedboxfull}
\newcommand{\markempty}{\qed}

\newtheorem{claim}[theorem]{Claim}
\newtheorem{fact}[theorem]{Fact}
\newtheorem{observation}{Observation}
\newtheorem{remark}{Remark}
\newtheorem{apptheorem}{Theorem}[section]
\newtheorem{appcorollary}[apptheorem]{Corollary}
\newtheorem{appproposition}[apptheorem]{Proposition}
\newtheorem{applemma}[apptheorem]{Lemma}
\newtheorem{appclaim}[apptheorem]{Claim}
\newtheorem{appfact}[apptheorem]{Fact}

\fancyhead{}

\title{Non-Uniformly Terminating Chase: Size and Complexity}

\author{Marco Calautti}
\affiliation{%
	\institution{University of Trento}
	\country{}
}
\email{marco.calautti@unitn.it}

\author{Georg Gottlob}
\affiliation{%
	\institution{University of Oxford}
	\country{}
}
\email{georg.gottlob@cs.ox.ac.uk}

\author{Andreas Pieris}
\affiliation{%
	\institution{University of Edinburgh \&}
	\country{}
}
\affiliation{%
	\institution{University of Cyprus}
	\country{}
}
\email{apieris@inf.ed.ac.uk}

\begin{abstract}
	The chase procedure, originally introduced for checking implication of database constraints, and later on used for computing data exchange solutions, has recently become a central algorithmic tool in rule-based ontological reasoning. In this context, a key problem is non-uniform chase termination: does the chase of a database w.r.t.~a rule-based ontology terminate? And if this is the case, what is the size of the result of the chase?
	We focus on guarded tuple-generating dependencies (TGDs), which form a robust rule-based ontology language, and study the above central questions for the semi-oblivious version of the chase.
	One of our main findings is that non-uniform semi-oblivious chase termination for guarded TGDs is feasible in polynomial time w.r.t.~the database, and the size of the result of the chase (whenever is finite) is linear w.r.t.~the database.
	Towards our results concerning non-uniform chase termination, we show that basic techniques such as simplification and linearization, originally introduced in the context of ontological query answering, can be safely applied to the chase termination problem.
\end{abstract}

\maketitle

\section{Introduction}\label{sec:introduction}

Nowadays we need to deal with data that is very large, heterogeneous,
distributed in different sources, and incomplete. This makes the task of extracting information from such data by means of queries very complex. 
At the same time, we have very large amounts of knowledge about the application domain of the data in the form of ontologies. This gave rise to a research field, recently coined as {\em knowledge-enriched data management}~\cite{manifesto}, that lies at the intersection of data management and knowledge representation and reasoning.
A major challenge for knowledge-enriched data management is to provide end users with flexible and integrated access to data by using the available knowledge about the underlying application domain.
{\em Ontology-based data access} (OBDA)~\cite{PLCD*08} has been proposed as a general paradigm for addressing the above challenge.
The main algorithmic task underlying OBDA is querying knowledge-enriched data, that is, during the query answering process we also need to take into account the inferred knowledge. This problem is also known as {\em ontological query answering}.

Typically, the ontologies employed in data-intensive applications such as OBDA are modeled via description logics, in particular, members of the DL-Lite~\cite{CDLL*07} and $\mathcal{EL}$~\cite{BaBL05} families, mainly due to their good computational properties when it comes to ontological query answering.
On the other hand, there is a consensus that rule-based ontologies, i.e., ontologies consisting of {\em tuple-generating dependencies} (TGDs) (a.k.a.~{\em existential rules}), are also well-suited for data-intensive applications since they allow us to conveniently deal with higher-arity relations that appear in standard relational databases.
In particular, {\em linear} and {\em guarded} TGDs strike a good balance between expressiveness and complexity that make them suitable ontology languages for data-intensive applications~\cite{CaGK13}.
Interestingly, the main members of the DL-Lite family (modulo some easily handled features) are special cases of linear TGDs (in fact, {\em simple linear} TGDs, where variables are not repeated in rule-bodies), while the main members of the $\mathcal{EL}$ family are (up to a certain normal form) special cases of guarded TGDs.

A prominent tool for studying rule-based ontological query answering is the \emph{chase procedure} (or simply chase). 
It takes as input a database $D$, and a rule-based ontology $\dep$, and, if it terminates, it computes a finite instance $D_\dep$ that is a {\em universal model} of $D$ and $\dep$, i.e., a model that can be homomorphically embedded into every other model of $D$ and $\dep$. 
This is why the chase is an important algorithmic tool for rule-based ontological query answering. 
And this is not only in theory. There are efficient implementations of the chase that allow us to solve ontological query answering by adopting a materialization-based approach.
It has been recently observed that for RAM-based implementations the {\em restricted} (a.k.a.~{\em standard}) version of the chase is the indicated tool~\cite{BKMMPST17,KrMR19}, while for 
%
RDBMS-based implementations 
the {\em semi-oblivious} chase is preferable~\cite{BKMMPST17}.
%
For being able, though, to employ existing chase implementations, we need a guarantee that the chase terminates.

This brings us to the chase termination problem (which is clearly parameterized by the version of the chase) that comes in two different variants: {\em uniform} and {\em non-uniform}.
The uniform variant takes as input a rule-based ontology $\dep$, and the question is whether the chase terminates for {\em every} database w.r.t.~$\dep$.
It is clear that whenever uniform chase termination is guaranteed, we can solve ontological query answering via materialization no matter how the database looks like. But even if uniform chase termination is not guaranteed, we can still rely in some cases on materialization, depending on the input database. This reveals the relevance of the non-uniform variant of the chase termination problem, which takes as input a database $D$ and a rule-based ontology $\dep$, and asks if the chase of $D$ w.r.t.~$\dep$ terminates.
It is well-known that, no matter which version of the chase we consider, both uniform and non-uniform chase termination is undecidable when we consider arbitrary rule-based ontologies~\cite{DeNR08,GoMa14}.
On the other hand, once we focus on well-behaved classes of TGDs, we have several positive results concerning the uniform variant of the chase termination problem.
The first such results were established for (simple) linear and guarded TGDs, and the semi-oblivious version of the chase~\cite{CaGP15}. There are also recent results for sticky TGDs (another well-behaved class that is inherently unguarded), and the semi-oblivious version of the chase~\cite{CaPi19}.
The restricted chase has been recently studied with linear TGDs~\cite{GoMP20KI,LMTU19}, as well as with guarded and sticky TGDs~\cite{GoMaP20}.

Although the uniform variant of chase termination has been extensively studied in the literature, the non-uniform one has largely remained unexplored. In this work, we concentrate on the semi-oblivious chase, and study the following questions: given a database $D$ and a (simple) linear or guarded rule-based ontology $\dep$: 
\begin{enumerate}
	\item What is the worst-case optimal size of the result of the chase of $D$ w.r.t.~$\dep$ (whenever is finite)?
	\item Can we decide whether the result of the chase of $D$ w.r.t.~$\dep$ is finite, and if yes, what is the exact complexity?
\end{enumerate}

\medskip

\noindent
\textbf{Summary of Contributions.} After illustrating the different nature of the non-uniformly terminating chase compared to the uniformly terminating one, we establish characterizations of non-uniform chase termination of the following form:

\medskip

\noindent \textsc{Main Characterizations.} 
Given a database $D$, and a (simple) linear or guarded rule-based ontology $\dep$, the following are equivalent:
	\begin{enumerate}
		\item The result of the chase of $D$ w.r.t.~$\dep$ is finite.
		\item The size of the result of the chase of $D$ w.r.t.~$\dep$ is bounded by $|D| \cdot f(\dep)$, for some computable function $f$ that maps ontologies to the natural numbers.
		\item $\dep$ enjoys a syntactic property (relative to the database $D$) that relies on a non-uniform version of weak-acyclicity.
	\end{enumerate}

\medskip

Interestingly, the result of the chase (whenever is finite) is linear w.r.t.~the size of the given database.
In the case of simple linear (resp., linear, guarded) rule-based ontologies, the function $f$ in the above characterizations is exponential (resp., double-exponential, triple-exponential) in the arity, and exponential (resp., exponential, double-exponential) in the number of predicates of the underlying schema. We further provide lower bounds on the size of the chase showing that the above upper bounds are worst-case optimal.

We then exploit the above characterizations (in fact, item (3)) to establish several complexity results for the problem in question that range from \text{AC}$_0$ to \textsc{2ExpTime}.
Among other results, we obtain that non-uniform semi-oblivious chase termination for linear (resp., guarded) rule-based ontologies is in \text{AC}$_0$ (resp., \textsc{PTime}-complete) in data complexity, i.e., when the ontology is considered fixed.

Towards the above characterizations, we establish results of independent interest concerning the basic techniques of simplification and linearization, originally introduced in the context of ontological query answering. Simplification eliminates the repetition of variables in rule-bodies, while linearization converts guarded rule-based ontologies into linear ones. We show that both techniques can be applied in the context of chase termination in the sense that they preserve the finiteness of the chase, and, more importantly, the depth of the terms occurring in the result of the chase.
\section{Preliminaries}\label{sec:preliminaries}

We consider the disjoint countably infinite sets $\ins{C}$, $\ins{N}$, and $\ins{V}$ of {\em constants}, {\em (labeled) nulls}, and {\em variables}, respectively. We refer to constants, nulls and variables as {\em terms}. For an integer $n > 0$, we may write $[n]$ for the set $\{1,\ldots,n\}$.

\medskip

\noindent 
\textbf{Relational Databases.} A {\em schema} $\ins{S}$ is a finite set of relation symbols (or predicates) with associated arity. We write $R/n$ to denote that $R$ has arity $n > 0$; we may also write $\arity{R}$ for the integer $n$.
A {\em (predicate) position} of $\ins{S}$ is a pair $(R,i)$, where $R/n \in \ins{S}$ and $i \in [n]$, that essentially identifies the $i$-th argument of $R$. We write $\pos{\ins{S}}$ for the set of positions of $\ins{S}$, that is, the set $\{(R,i) \mid R/n \in \ins{S} \text{ and } i \in [n]\}$.
An {\em atom} over $\ins{S}$ is an expression of the form $R(\bar t)$, where $R/n \in \ins{S}$ and $\bar t$ is an $n$-tuple of terms. A {\em fact} is an atom whose arguments consist only of constants.
For a variable $x$ in $\bar t = (t_1,\ldots,t_n)$, let $\posvar{R(\bar t)}{x} = \{(R,i) \mid t_i = x\}$. 
We write $\var{R(\bar t)}$ for the set of variables in $\bar t$. The notations $\posvar{\cdot}{x}$ and $\var{\cdot}$ extend to sets of atoms.
An {\em instance} over $\ins{S}$ is a (possibly infinite) set of atoms over $\ins{S}$ with constants and nulls. A {\em database} over $\ins{S}$ is a finite set of facts over $\ins{S}$. The {\em active domain} of an instance $I$, denoted $\adom{I}$, is the set of terms (constants and nulls) occurring in $I$. 

\medskip

\noindent
\textbf{Substitutions and Homomorphisms.}
A {\em substitution} from a set of terms $T$ to a set of terms $T'$ is a function $h : T \ra T'$. Henceforth, we treat a substitution $h$ as the set of mappings $\{t \mapsto h(t) \mid t \in T\}$.
The restriction of $h$ to a subset $S$ of $T$, denoted $h_{|S}$, is the substitution $\{t \mapsto h(t) \mid t \in S\}$.
A {\em homomorphism} from a set of atoms $A$ to a set of atoms $B$ is a substitution $h$ from the set of terms in $A$ to the set of terms in $B$ such that $h$ is the identity on $\ins{C}$, and $R(t_1,\ldots,t_n) \in A$ implies $h(R(t_1,\ldots,t_n)) =  R(h(t_1),\ldots,h(t_n)) \in B$.

\medskip

\noindent
\textbf{Tuple-Generating Dependencies.} A {\em tuple-generating dependency} (TGD) $\sigma$ is a (constant-free) first-order sentence of the form
$
\forall \bar x \forall \bar y \left(\phi(\bar x,\bar y) \ra \exists \bar z\, \psi(\bar x,\bar z)\right),
$
where $\bar x, \bar y$ and $\bar z$ are tuples of variables of $\ins{V}$, and $\phi(\bar x,\bar y)$ and $\psi(\bar x,\bar z)$ are non-empty conjunctions of atoms that mention only variables from $\bar x \cup \bar y$ and $\bar x \cup \bar z$, respectively. Note that, by abuse of notation, we may treat a tuple of variables as a set of variables.
We write $\sigma$ as $\phi(\bar x,\bar y) \ra \exists \bar z\, \psi(\bar x,\bar z)$, and use comma instead of $\wedge$ for joining atoms. We refer to $\phi(\bar x,\bar y)$ and $\psi(\bar x,\bar z)$ as the {\em body} and {\em head} of $\sigma$, denoted $\body{\sigma}$ and $\head{\sigma}$, respectively.
The {\em frontier} of the TGD $\sigma$, denoted $\fr{\sigma}$, is the set of variables $\bar x$, i.e., the variables that appear both in the body and the head of $\sigma$. 
The {\em schema} of a set $\dep$ of TGDs, denoted $\sch{\dep}$, is the set of predicates occurring in $\dep$, and we write $\arity{\dep}$ for the maximum arity over all those predicates. 
We assume, w.l.o.g., that no two TGDs of $\dep$ share a variable, and we let $||\dep|| = |\atoms{\dep}| \cdot |\sch{\dep}| \cdot \arity{\dep}$, where $\atoms{\dep}$ is the set of atoms occurring in the TGDs of $\dep$.
An instance $I$ satisfies a TGD $\sigma$ as the one above, written $I \models \sigma$, if whenever there exists a homomorphism $h$ from $\phi(\bar x, \bar y)$ to $I$, then there is $h' \supseteq h_{|\bar x}$ that is a homomorphism from $\psi(\bar x,\bar z)$ to $I$; we may treat a conjunction of atoms as a set of atoms. The instance $I$ satisfies a set $\dep$ of TGDs, written $I \models \dep$, if $I \models \sigma$ for each $\sigma \in \dep$.

\medskip

\noindent
\textbf{Guardedness.} A TGD $\sigma$ is called \emph{guarded} if there exists an atom in $\body{\sigma}$ that contains (or ``guards'') all the variables in $\body{\sigma}$. Conventionally, the leftmost such an atom in $\body{\sigma}$ is the {\em guard} of $\sigma$, denoted $\guard{\sigma}$. The class of guarded TGDs, denoted
$\class{G}$, is defined as the family of all finite sets of
guarded TGDs.
A TGD is called {\em linear} if it has only one atom in its body, and the corresponding class is denoted $\class{L}$. We further call a linear TGD {\em simple} if no variable occurs more than once in its body-atom, and the corresponding class is denoted $\class{SL}$. 
It is clear that $\class{SL} \subsetneq \class{L} \subsetneq \class{G}$.
\section{The Semi-Oblivious Chase Procedure}\label{sec:chase-procedure}
The semi-oblivious chase procedure (or simply chase) takes as input a database $D$ and a set $\dep$ of TGDs, and constructs an instance that contains $D$ and satisfies $\dep$.
Central notions in this context are those of trigger, active trigger, and trigger application.

\begin{definition}
Given a set $\dep$ of TGDs, and an instance $I$, a {\em trigger} for $\dep$  on $I$ is a pair $(\sigma,h)$, where $\sigma \in \dep$ and $h$ is a homomorphism from $\body{\sigma}$ to $I$.
The {\em result} of $(\sigma,h)$, denoted $\result{\sigma}{h}$, is the set $\mu(\head{\sigma})$, where $\mu : \var{\head{\sigma}} \ra \ins{C} \cup \ins{N}$ is defined as follows:
%
\[
\mu(x)\
=\ \left\{
\begin{array}{ll}
	h(x) & \quad \text{if } x \in \fr{\sigma}\\
	&\\
	\bot_{\sigma,h_{|\fr{\sigma}}}^{x} & \quad \text{otherwise}
\end{array} \right.
\]
where $\bot_{\sigma,h_{|\fr{\sigma}}}^{x}$ is a null value from $\ins{N}$.
The trigger $(\sigma,h)$ is {\em active} if $\result{\sigma}{h} \not\subseteq I$.
The {\em application} of $(\sigma,h)$ to $I$ returns the instance $J = I \cup \result{\sigma}{h}$, and is denoted as $I \app{\sigma}{h} J$.
\hfill\markfull
\end{definition}

Observe that in the definition of $\result{\sigma}{h}$ above each existentially quantified variable $z$ of $\head{\sigma}$ is mapped by $\mu$ to a null value of $\ins{N}$ whose name is uniquely determined by the trigger $(\sigma,h)$ and $z$ itself. This means that, given a trigger $(\sigma,h)$, we can unambiguously write down the set of atoms 
$\result{\sigma}{h}$.


The main idea of the chase is, starting from a database $D$, to exhaustively apply active triggers for the given set $\dep$ of TGDs on the instance constructed so far. This is formalized via the notion of (semi-oblivious) chase derivation, which can be finite or infinite.

\begin{definition}
	Consider a database $D$, and a set $\dep$ of TGDs.
\begin{itemize}
\item A finite sequence $(I_i)_{0 \leq i \leq n}$ of instances, with $D = I_0$ and $n \geq 0$, is a {\em chase derivation} of $D$ w.r.t.~$\dep$ if, for each $i \in \{0,\ldots,n-1\}$, there is an active trigger $(\sigma,h)$ for $\dep$ on $I_i$ with $I_i \app{\sigma}{h} I_{i+1}$, and there is no active trigger for $\dep$ on $I_n$. The {\em result} of such a chase derivation is the instance $I_n$.

\item An infinite sequence $(I_i)_{i \geq 0}$ of instances, with $D = I_0$, is a {\em chase derivation} of $D$ w.r.t.~$\dep$ if, for each $i \geq 0$, there is an active trigger $(\sigma,h)$ for $\dep$ on $I_i$ such that $I_i \app{\sigma}{h} I_{i+1}$. Moreover, $(I_i)_{i \geq 0}$ is {\em fair} if, for each $i \geq 0$, and for every active trigger $(\sigma,h)$ for $\dep$ on $I_i$, there exists $j > i$ such that $(\sigma,h)$ is not an active trigger for $\dep$ on $I_j$. 
The {\em result} of such a chase derivation is the instance $\bigcup_{i \geq 0} \, I_i$.
\end{itemize}
%
A chase derivation is {\em valid} if it is finite, or infinite and fair.  \hfill\markfull
\end{definition}

Let us stress that infinite but unfair chase derivations are not considered as valid ones since they do not serve the main purpose of the chase, that is, to build an instance that satisfies the given set of TGDs. Indeed, given the set $\dep$ consisting of the TGDs
\[
\sigma\ =\ R(x,y) \ra \exists z \, R(y,z) \qquad \sigma'\ =\ R(x,y) \ra P(x,y),
\]
the result of the unfair chase derivation of $D = \{R(a,b)\}$ w.r.t.~$\dep$ that involves only triggers of the form $(\sigma,\cdot)$, i.e., only the TGD $\sigma$ is used, does not satisfy $\sigma'$, and thus, it does not satisfy $\dep$.

%

\medskip

\noindent
\textbf{Non-Uniform Chase Termination.}
A valid chase derivation may be infinite even for very simple settings: it is easy to see that the only chase derivation of the database $D = \{R(a,b)\}$ w.r.t.~the singleton set of TGDs $\dep = \{R(x,y) \ra \exists z \, R(y,z)\}$ is infinite.
The questions that come up are, given a database $D$ and a set $\dep$ of TGDs: 
\begin{enumerate}
	\item What is the worst-case optimal size (that is, the cardinality) of the result of a finite chase derivation of $D$ w.r.t.~$\dep$?
	\item Can we decide whether all or some valid chase derivations of $D$ w.r.t.~$\dep$ are finite? 
\end{enumerate}
To properly formalize the above questions, we need to recall some central classes of sets of TGDs, parameterized by a database $D$:
\[
\class{CT}^{\forall}_{D}\ =\ \left\{\dep \mid \text{ {\em every} valid chase derivation of } D \text{ w.r.t.~} \dep \text{ is finite}\right\}\
\]
and
\[
\begin{array}{rcl}
\class{CT}^{\exists}_{D} &=& \left\{\dep\ ~\left|
\begin{array}{c}
\text{there {\em exists} a valid chase derivation of }\\
D \text{ w.r.t.~} \dep \text{ that is finite}
\end{array} \right.\right\}.
\end{array}
\]
It is well-known from~\cite{GrOn18} that, for every database $D$, $\class{CT}^{\forall}_{D} = \class{CT}^{\exists}_{D}$; henceforth, we simply write $\class{CT}_{D}$ for $\class{CT}^{\forall}_{D}$ and $\class{CT}^{\exists}_{D}$.
Furthermore, it is known that for every database $D$ and set $\dep$ of TGDs, two distinct valid chase derivations of $D$ w.r.t.~$\dep$ have always the same result, which we denote by $\chase{D}{\dep}$. Consequently,
\[
\class{CT}_{D}\ =\ \left\{\dep \mid \text{the instance } \chase{D}{\dep} \text{ is finite}\right\}.
\]
Therefore, question (1) above essentially asks for a worst-case optimal upper bound on the cardinality of the instance $\chase{D}{\dep}$, assuming that $\dep \in \class{CT}_D$, that is, $\chase{D}{\dep}$ is finite.
On the other hand, question (2) above corresponds to the following decision problem, parameterized by a class $\class{C}$ of sets of TGDs:

\medskip

\begin{center}
	\fbox{
		\begin{tabular}{ll}
			{\small PROBLEM} : & $\mathsf{ChTrm}(\class{C})$
			\\
			{\small INPUT} : & A database $D$ and a set $\dep \in \class{C}$.
			\\
			{\small QUESTION} : &  Is it the case that $\dep \in \class{CT}_{D}$?
	\end{tabular}}
\end{center}




\medskip

\noindent
\textbf{Our Goal.}
With $\class{TGD}$ being the class of arbitrary sets of TGDs, we know that $\mathsf{ChTrm}(\class{TGD})$ is undecidable. 
This was shown in~\cite{DeNR08} for the restricted chase, but it was observed in~\cite{Marn09} that the same proof applies to the semi-oblivious version of the chase.
In view of the fact that the set of TGDs employed in the undecidability proof of~\cite{DeNR08} is far from being guarded, we are interested in studying the above questions for guarded TGDs, and subclasses thereof. In other words, we are focussing on the following questions, for $\class{C} \in \{\class{SL},\class{L},\class{G}\}$: 
\begin{itemize}
	\item Given a database $D$, and a set of TGDs $\dep \in \class{C} \cap \class{CT}_D$, what is the worst-case optimal size of the instance $\chase{D}{\dep}$?
	\item Is $\mathsf{ChTrm}(\class{C})$ decidable, and if yes, what is the complexity?
\end{itemize}
The rest of the paper is devoted to providing answers to the above research questions. To this end, we are going to establish, for each class $\class{C} \in \{\class{SL},\class{L},\class{G}\}$, a characterization of non-uniform chase termination of the following form, which is of independent interest:

\medskip

\noindent \textsc{Target Characterization.} 
{\em Consider a database $D$, and a set $\dep \in \class{C}$ of TGDs. The following are equivalent:
\begin{enumerate}
	\item $\dep \in \class{CT}_D$.
	\item $|\chase{D}{\dep}| \leq |D| \cdot f_\class{C}(\dep)$, for some computable function $f_\class{C} : \class{C} \ra \mathbb{N}$ (as usual, $\mathbb{N}$ denotes the set of natural numbers).
	\item $\dep$ enjoys a syntactic property (relative to the database $D$) that relies on a non-uniform version of weak-acyclicity.
\end{enumerate}
}

\medskip

Notice that, for all the classes of TGDs in question, the size of the chase instance is linear w.r.t. the size of the given database.
We further complement the above characterization 
with a lower bound showing that the provided upper bound is worst-case optimal.

Once we have the above results in place, for each class $\class{C} \in \{\class{SL},\class{L},\class{G}\}$ of TGDs, we will immediately get an answer concerning the worst-case optimal size of the chase instance.
%
Concerning the decidability of $\mathsf{ChTrm}(\class{C})$, it is clear that item (2) of the above characterization will provide a simple decision
procedure: given a database $D$ and a set $\dep \in \class{C}$, simply construct the instance $\chase{D}{\dep}$, and if $|\chase{D}{\dep}|$ exceeds the value $|D| \cdot f_\class{C}(\dep)$, then reject; otherwise, accept. However, in most of the cases, this naive approach will not provide worst-case optimal complexity upper bounds.
Towards optimal complexity bounds, we are going to exploit item (3) of the above characterization. Indeed, as we shall see, in all the cases, the procedure of checking whether the syntactic property in question holds provides optimal complexity bounds for $\mathsf{ChTrm}(\class{C})$.


\OMIT{
\subsection{Some Useful Results}
Before we proceed further, let us present a couple of results that will be very useful for our analysis.


\paragraph{Fairness.} The first one concerns the fairness condition of the chase. As one might expect, we are going to focus on the complement of $\mathsf{CT}(\class{C})$, where $\class{C} \in \{\class{G},\class{L},\class{SL}\}$, and pinpoint the complexity of the following problem: given a database $D$ and a set $\dep \in \class{C}$, is there a valid infinite chase derivation $\delta$ of $D$ w.r.t.~$\dep$.
However, as noted in~\cite{CaGP15}, where the uniform chase termination problem has been studied, one of the main difficulties is to ensure that $\delta$ indeed enjoys the fairness condition. Interestingly, as shown in~\cite{CaGP15}, we can completely neglect the fairness condition since the existence of a (possibly unfair) infinite $\star$-chase derivation of some database w.r.t.~$\dep$ implies the existence of a fair one.

\begin{proposition}\label{pro:fairness}
	Consider a database $D$, and a set $\dep$ of TGDs. The following are equivalent:
	\begin{enumerate}
		\item There is a valid infinite $\star$-chase derivation of $D$ w.r.t.~$\dep$.
		\item There is an infinite $\star$-chase derivation of $D$ w.r.t.~$\dep$
	\end{enumerate}
\end{proposition}

Note that the above result holds for arbitrary sets of TGDs, not necessarily guarded.

\OMIT{
\paragraph{Common schema.}
Let $\star \in \{\obl,\sobl\}$, and consider a database $D$ and a set $\dep$ of TGDs. Then, note that $\starchase{D}{\dep}$ is finite iff $\starchase{D'}{\dep}$ is finite, for every database $D'$ obtained from $D$, where every atom with predicate not occurring in $\sch{\dep}$ is removed. Thus, we can always assume w.l.o.g.\, that a database $D$ only contains atoms with predicates occurring in the given set $\dep$ of TGDs. Hereafter, when a database $D$ and a set $\dep$ of TGDs are considered, we always assume, unless specified otherwise, that every predicate occurring in $D$ also occurs in $\sch{\dep}$.

\paragraph{Reduction between chase varations.} Since the semi-oblivious chase is a refined version of the oblivious chase, it is not surprising that $\ctdpr{\obl}(\class{TGD})$ can be reduced to $\ctdpr{\sobl}(\class{TGD})$.
This relies on a very simple construction, known as enrichment~\cite{GrOn18}. Formally, the \emph{enrichment} of a set $\dep$ of TGDs, denoted $\oblrew{\dep}$, is the set of TGDs obtained by replacing each TGD $\sigma \in \dep$ of the form $\phi(\bar x,\bar y) \ra \exists \bar z\ \psi(\bar x,\bar z)$ with the TGD
\[
\phi(\bar x,\bar y) \ra \exists \bar z\ \psi(\bar x,\bar z), \text{\rm Aux}_{\sigma}(\bar x, \bar y),
\]
where $\text{\rm Aux}_{\sigma}$ is an auxiliary predicate of arity $(|\bar x|+|\bar y|)$ not occurring in $\sch{\dep}$.
It is an easy exercise to show that the notion of enrichment provides the desired reduction from $\ctdpr{\obl}(\class{TGD})$ to $\ctdpr{\sobl}(\class{TGD})$. More precisely:

\begin{lemma}\label{lem:obl-sobl-reduction}
	For every database $D$, and set $\dep$ of TGDs, the following hold:
	\begin{itemize}
		\item $\dep \in \ctd{\obl}{D}$ iff $\oblrew{\dep} \in \ctd{\sobl}{D}$, and 
		\item $|\ochase{D}{\dep}| \le |\sochase{D}{\oblrew{\dep}}| \le |\ochase{D}{\dep}| + N$,
	\end{itemize}
	where $N$ is the number of nulls in $\ochase{D}{\dep}$.
\end{lemma}

Thus, in what follows, we will focus on the semi-oblivious chase. We will then exploit Lemma~\ref{lem:obl-sobl-reduction} to transfer our results to the oblivious chase.
}

\paragraph{Bounding the Chase Size.} The second result establishes a generic upper bound on the number of atoms occurring in the instance $\chase{D}{\dep}$, for some database $D$ and set $\dep \in \class{G}$ of TGDs, assuming that $\dep \in \class{CT}_D$, based on the depth of the terms in $\chase{D}{\dep}$.
The {\em depth (w.r.t.~$D$ and $\dep$)} of a term $t \in \adom{\chase{D}{\dep}}$ is inductively defined as follows: 
\begin{itemize}
\item if $t \in \adom{D}$, then $\depth{t} = 0$,
\item if $t$ is a null of the form $\bot_{\sigma,h_{|\fr{\sigma}}}^{z}$ and $\fr{\sigma} = \emptyset$ (which means that $h_{|\fr{\sigma}} = \emptyset$), then $\depth{t} = 1$, and 
\item if $t$ is a null of the form $\bot_{\sigma,h_{|\fr{\sigma}}}^{z}$ and $\fr{\sigma} \neq \emptyset$, then
$
\depth{t} = \max_{x \in \fr{\sigma}} \{\depth{h(x)}\} +1.
$
\end{itemize}
It is clear that, if $\chase{D}{\dep}$ is finite, then there exists an integer $k_{D,\dep} \geq 0$ such that, for each $t \in \adom{\chase{D}{\dep}}$, $\depth{t} \leq k_{D,\dep}$.
We are now ready to provide the generic bound on the number of atoms occurring in $\chase{D}{\dep}$ claimed above.
Let $\mathsf{chsize}(\cdot,\cdot)$ be the function from pairs consisting of a database and a set of TGDs to the natural numbers defined as expected, that is, $\mathsf{chsize}(D,\dep) = |\chase{D}{\dep}|$.

\begin{proposition}\label{pro:generic-bound}
	Consider a database $D$, and a set $\dep \in \class{G}$ of TGDs. If $\dep \in \class{CT}_{D}$, then
	\[
	\chasesize{D}{\dep}\ \in\ O\left(|D|  \cdot ||\dep||^{2 \cdot \arity{\dep} \cdot (d+1)} \cdot d\right),
	\]
	where $d = \max_{t \in \adom{\chase{D}{\dep}}} \{\depth{t}\}$.
\end{proposition}

\begin{proposition}\label{pro:generic-bound}
	Consider a database $D$, and a set $\dep \in \class{G}$ of TGDs. If $\dep \in \class{CT}_{D}$, then
	\[
	\chasesize{D}{\dep}\ \in\ |D|  \cdot ||\dep||^{O(\arity{\dep} \cdot d)} \cdot (d+1),
	\]
	where $d = \max_{t \in \adom{\chase{D}{\dep}}} \{\depth{t}\}$.
\end{proposition}

\begin{proposition}\label{pro:generic-bound}
	Consider a database $D$, and a set $\dep \in \class{G}$ of TGDs. If $\dep \in \class{CT}_{D}$, then
	\[
	\chasesize{D}{\dep}\ \in\ |D|  \cdot ||\dep||^{O(d)} \cdot \arity{\dep}^{O(\arity{\dep} \cdot d)} \cdot (d+1),
	\]
	where $d = \max_{t \in \adom{\chase{D}{\dep}}} \{\depth{t}\}$.
\end{proposition}

Note that the above result, whose proof can be found in the appendix, heavily relies on guardedness, and it does not hold for arbitrary sets of TGDs.
The way that Proposition~\ref{pro:generic-bound} is stated, it gives the impression that the number of atoms in $\chase{D}{\dep}$, whenever $\dep \in \class{G} \cap \class{CT}_D$, is at least exponential in the database $D$, since the depth of a null may depend on $D$, which is of course an undesirable outcome. Interestingly, as we shall see in the next sections, due to guardedness, the depth of a term depends only on the set of TGDs, which in turn implies that $|\chase{D}{\dep}|$ is always polynomial in $|D|$.

\OMIT{
\paragraph{Depth.} We now recall a useful connection between the finiteness of the (semi-)oblivious chase and the so-called \emph{depth} of the nulls occurring in it.

As discussed above, we focus only one the $\sobl$-chase.
Consider a database $D$, a set $\dep$ of TGDs and let $t \in \adom{\sochase{D}{\dep}}$ be a term occurring in the chase of $D$ w.r.t.\ $\dep$.
The \emph{depth} of $t$ (w.r.t.\ $D$ and $\dep$), is the integer defined as follows. If $t \in \adom{D}$, then $\depth{t} = 0$. Otherwise, let $t$ be a null of the form $\nullobl{x}{\sigma}{h}$, for some TGD $\sigma \in \dep$, existential variable $x$ of $\sigma$, and mapping $h$. Then, we define
$$\depth{t} = \max\limits_{x \in \fr{\sigma}} \{ \depth{h(x)} \} + 1.$$
Note that if $\fr{\sigma} = \emptyset$, we assume by convention that $\depth{t} = 1$.
If $\sochase{D}{\dep}$ is finite, note that there exists an integer that is the maximum depth among the terms in $\sochase{D}{\dep}$. We denote such a depth as $\md{D,\dep}$. We extend the notion of depth to atoms as follows. Let $\alpha \in \sochase{D}{\dep}$. The depth of $\alpha$ (w.r.t.\ $D$ and $\dep$), denoted $\depth{\alpha}$, is the integer
$\depth{\alpha} = \max_{t \in \adom{\alpha}} \{\depth{t} \}$.\footnote{We remark that for guarded TGDs, the depth of an atom does not necessarily coincide with the ``depth'' such an atom has in a tree of the guarded chase forest.}

\begin{proposition}\label{prop:depth}
	For every database $D$ and every set $\dep$ of TGDs, $\sochase{D}{\dep}$ is finite iff the depth of all its terms is bounded.
\end{proposition}
\begin{proof}
	TODO Do we need a proof?
\end{proof}
}

}
\section{Uniformly vs. Non-Uniformly Terminating Chase}\label{sec:comparison}

The first step of our analysis towards the characterizations described above, will be to provide a generic upper bound on the size of the instance $\chase{D}{\dep}$ (whenever is finite), for a database $D$ and a set $\dep$ of guarded TGDs, of the form $|D| \cdot k$, where $k$ depends only on $\dep$ and the depth of the terms occurring in $\chase{D}{\dep}$ (the notion of depth is defined below).
We will then provide, for each class $\class{C} \in \{\class{SL},\class{L},\class{G}\}$ of TGDs, a database-independent upper bound on the depth of the terms that occur in a chase instance.
By combining the above bounds, we will eventually get the desired upper bound of the form $|D| \cdot f_\class{C}(\dep)$ on the size of $\chase{D}{\dep}$ (whenever is finite), for a database $D$ and a set of TGDs $\dep \in \class{C}$.

At this point, one may think that the above bounds can be obtained by merely adapting existing results for arbitrary (not necessarily guarded) TGDs concerning the uniformly terminating chase. Therefore, before we proceed with our analysis, we would like to stress the different nature of the non-uniformly terminating chase.

\medskip

\noindent 
\textbf{Bounded Chase Size.}
We know from~\cite{Marn09} that, if we focus on arbitrary sets of TGDs that ensure the {\em uniform} termination of the chase, then we can bound the size of the chase instance via a uniform function over the input database. More precisely, with
\[
\class{CT}\ =\ \left\{\dep \in \class{TGD} \mid \text{for {\em every} database } D, \chase{D}{\dep} \text{ is finite}\right\}
\]
the following result is implicit in~\cite{Marn09}:

\begin{theorem}\label{the:marnette}
	For a set $\dep \in \class{CT}$, there is a computable function $f_\dep$ from databases to $\mathbb{N}$ s.t., for every database $D$, $|\chase{D}{\dep}| \leq f_\dep(D)$.
\end{theorem}

Let us stress that the function $f_\dep$ provided by Theorem~\ref{the:marnette} is actually a polynomial, which in turn implies that, for a set $\dep \in \class{CT}$, $|\chase{D}{\dep}|$ is polynomial w.r.t.~the size of $D$, for every database $D$.
One may think that a version of Theorem~\ref{the:marnette} for the non-uniform case, where guardedness does not play any crucial role, can be easily obtained by adapting the proof of Theorem~\ref{the:marnette}. 
In other words, one may expect a result of the following form: for a set of (not necessarily guarded) TGDs $\dep$, there is a computable function $f_\dep$ from databases to $\mathbb{N}$ such that, for every database $D$ with $\dep \in \class{CT}_D$, $|\chase{D}{\dep}| \leq f_\dep(D)$.
We proceed to show that this is not true.

\begin{proposition}\label{pro:marnette-non-uniform}
	There exists a set $\dep$ of TGDs such that the following holds: for every computable function $f_\dep$ from databases to $\mathbb{N}$, there exists a database $D$ with $\dep \in \class{CT}_D$ such that $|\chase{D}{\dep}| > f_\dep(D)$.
\end{proposition}


To show the above proposition, we first strengthen the undecidability result for $\mathsf{ChTrm}(\class{TGD})$ from~\cite{DeNR08} by showing that it remains undecidable even in data complexity. More precisely, we show that there exists a set $\dep^\star$ of TGDs such that the problem

\medskip

\begin{center}
	\fbox{
		\begin{tabular}{ll}
			{\small PROBLEM} : & $\mathsf{ChTrm}(\dep^\star)$
			\\
			{\small INPUT} : & A database $D$.
			\\
			{\small QUESTION} : &  Is it the case that $\dep^\star \in \class{CT}_{D}$?
	\end{tabular}}
\end{center}

\medskip

\noindent is undecidable; the proof can be found in the appendix.
Now, by contradiction, assume that there is a computable function $f_{\dep^\star}$ from databases to $\mathbb{N}$ such that, for every database $D$ with $\dep^\star \in \class{CT}_D$, $|\chase{D}{\dep^\star}| \leq f_{\dep^\star}(D)$. This implies that, for every database $D$, $\dep^\star \in \class{CT}_D$ iff $|\chase{D}{\dep^\star}| \leq f_{\dep^\star}(D)$. Since $f_{\dep^\star}$ is computable, we conclude that $\mathsf{ChTrm}(\dep^\star)$ is decidable, which is a contradiction.

\medskip

\noindent 
\textbf{Bounded Term Depth.} As said above, a key notion that will play a crucial role in our analysis is the depth of a term occurring in a chase instance. The formal definition follows.

\begin{definition}\label{def:term-depth}
Consider a database $D$, and a set $\dep$ of TGDs. For a term $t \in \adom{\chase{D}{\dep}}$, the {\em depth of $t$}, denoted $\depth{t}$, is inductively defined as follows:
\begin{itemize}
	\item if $t$ is a constant, then $\depth{t} = 0$, and
	\item if $t$ is a null value of the form $\bot_{\sigma,h}^{z}$, then $\depth{t} = 1 + \max \{\{\depth{h(x)} \mid x \in \fr{\sigma}\} \cup \{0\}\}$.
\end{itemize}
We write $\mathsf{maxdepth}(D,\dep)$ for the maximal depth over all terms of $\adom{\chase{D}{\dep}}$, that is, $\max_{t \in \adom{\chase{D}{\dep}}} \{\depth{t}\}$; in case $\chase{D}{\dep}$ is infinite, then $\mathsf{maxdepth}(D,\dep) = \infty$. \hfill\markfull
\end{definition}

An interesting result in the case of uniformly terminating chase, which provides a database-independent bound on the depth of terms in a chase instance, has been recently established in~\cite{BLMTUG19}:

\begin{theorem}\label{the:term-depth-uniform}
	For every set $\dep \in \class{CT}$, there exists an integer $k_\dep \geq 0$ such that, for every database $D$, $\mathsf{maxdepth}(D,\dep) \leq k_\dep$.
\end{theorem}

Again, one may be tempted to think that a version of Theorem~\ref{the:term-depth-uniform} for the non-uniform case, where guardedness does not play any crucial role, can be obtained by adapting the proof of Theorem~\ref{the:term-depth-uniform}. It is not difficult to show though that this is not the case.

\begin{proposition}\label{pro:unbounded-depth}
	There is a set $\dep$ of TGDs, and family of databases $\{D_n\}_{n>1}$ with $n = |D_n|$ and $\dep \in \class{CT}_{D_n}$ s.t. $\mathsf{maxdepth}(D_n,\dep) = n-1$.
\end{proposition}

	To prove the above result, consider the family $\{D_n\}_{n>1}$, where
	\[
	D_n\ =\ \{P(a_1,b,b),R(a_1,a_2),R(a_2,a_3),\ldots,R(a_{n-1},a_n)\},
	\]
	and the singleton set of TGDs $\dep$ consisting of
	\[
	\sigma\ =\ R(x,y),P(x,z,v)\ \ra\ \exists w \, P(y,w,z).
	\]
	It is clear that $\dep \in \class{CT}_{D_n}$, and that $\mathsf{maxdepth}(D_n,\dep) = n-1$.
	\OMIT{	
	In particular, $\chase{D_n}{\dep}$ is
	\[
	D_n\ \cup\ \{P(a_2,\bot_1,b)\}\ \cup\ \{P(a_i,\bot_{i-1},\bot_{i-2})\}_{i \in \{3,\ldots,n\}}
	\]
	where
	\begin{eqnarray*}
		\bot_1 &=& \bot^{w}_{\sigma,\{y \mapsto a_2, z \mapsto b\}}\\
		\bot_2 &=& \bot^{w}_{\sigma,\{y \mapsto a_3, z \mapsto \bot_1\}}\\
		\bot_j &=& \bot^{w}_{\sigma,\{y \mapsto a_{j+1}, z \mapsto \bot_{j-1}\}}
	\end{eqnarray*}
	for each $j \in \{3,\ldots,n-1\}$. It is clear that $\depth{\bot_j} = j$, and thus, $\mathsf{maxdepth}(D_n,\dep) = n-1$, as needed.}
%
Note that this does not contradict Theorem~\ref{the:term-depth-uniform} since $\dep \not\in\class{CT}$; indeed, $\chase{D}{\dep}$, where $D = \{P(a,a,a),R(a,a)\}$, is infinite.

Propositions~\ref{pro:marnette-non-uniform} and~\ref{pro:unbounded-depth} illustrate the different nature of the non-uniformly terminating chase compared to the uniformly terminating one. As we shall see, our analysis reveals that there are versions of Theorems~\ref{the:marnette} and~\ref{the:term-depth-uniform} for the non-uniform case, providing that we concentrate on guarded TGDs, which rely on techniques and proofs different than the ones used for Theorems~\ref{the:marnette} and~\ref{the:term-depth-uniform}.
\section{Bounding the Chase Size}\label{sec:generic-bound}

We proceed to provide a generic upper bound on the size of the instance $\chase{D}{\dep}$ (whenever is finite), for some database $D$ and set $\dep \in \class{G}$ of TGDs, based on $\mathsf{maxdepth}(D,\dep)$. To this end, we first need to recall some auxiliary notions.

\medskip

\noindent
\textbf{Guarded Chase Forest and Trees.}
Consider a valid chase derivation $\delta = (I_i)_{i \geq 0}$ of $D$ w.r.t.~$\dep$ with $I_i \app{\sigma_i}{h_i} I_{i+1}$ for each $i \geq 0$, which means that $I_{i+1} = I_i \cup \result{\sigma_i}{h_i}$. The {\em guarded chase forest of $\delta$}, denoted $\gforest{\delta}$, is the directed graph $(V,E)$, where $V = \bigcup_{i \geq 0} I_i$, and $(\alpha,\beta) \in E$ iff there is $i \geq 0$ such that $\alpha = h_i(\guard{\sigma_i})$ and $\beta \in I_{i+1} \setminus I_i$.
It is easy to verify that $\gforest{\delta}$ is indeed a forest consisting of directed trees rooted at the atoms of $D$. For an atom $\alpha \in D$, let $\gtree{\delta,\alpha}$ be the tree of $\gforest{\delta}$ rooted at $\alpha$.
In what follows, by abuse of notation, we may treat $\gforest{\delta}$ and $\gtree{\delta,\alpha}$ as the sets of their nodes, namely as sets of atoms.

\medskip

\noindent
\textbf{The Generic Bound.}
The notion of depth defined for terms can be transferred to atoms: given an atom $\alpha = R(t_1,\ldots,t_n) \in \chase{D}{\dep}$, the {\em depth of $\alpha$}, denoted $\depth{\alpha}$, is defined as $\max_{i \in [n]} \{\depth{t_i}\}$.
We show a lemma that provides an upper bound on the number of atoms of a certain depth occurring in $\gtree{\delta,\alpha}$.
We define the set
\[
\mathsf{gtree}^{i}(\delta,\alpha)\ =\ \{\beta \in \gtree{\delta,\alpha} \mid \depth{\beta} = i\}, 
\]
that is, the set of atoms of $\gtree{\delta,\alpha}$ of depth $i \geq 0$. 
We can then show the following key technical lemma:

\begin{lemma}\label{lem:depth-bound}
	Consider a database $D$, and a set $\dep \in \class{G}$ of TGDs. Let $\delta$ be a valid chase derivation of $D$ w.r.t.~$\dep$. For each $\alpha \in D$ and $i \geq 0$,
	\[
	|\mathsf{gtree}^{i}(\delta,\alpha)|\ \leq\ ||\dep||^{2 \cdot \arity{\dep} \cdot (i+1)}.
	\]
\end{lemma}

The proof of the above result, which can be found in the appendix,  is by induction on the depth $i \geq 0$.
We can now show the following:

\begin{proposition}\label{pro:generic-bound}
	Consider a database $D$, and a set $\dep \in \class{G} \cap \class{CT}_{D}$ of TGDs. With $d = \mathsf{maxdepth}(D,\dep)$, it holds that
	\[
	|\chase{D}{\dep}|\ \leq\ |D|  \cdot (d+1) \cdot ||\dep||^{2 \cdot \arity{\dep} \cdot (d+1)}.
	\]
\end{proposition}

	To prove the above result, consider a valid chase derivation $\delta$ of $D$ w.r.t~$\dep$. It is straightforward to see that, for an atom $\alpha \in D$, 
	\[
	\gtree{\delta,\alpha}\ =\ \bigcup_{i=0}^{d} \mathsf{gtree}^{i}(\delta,\alpha).
	\]
	Therefore, by Lemma~\ref{lem:depth-bound}, 
	$|\gtree{\delta,\alpha}| \leq (d+1) \cdot ||\dep||^{2 \cdot \arity{\dep} \cdot (d+1)}$. 
	It is clear that $\chase{D}{\dep}$ coincides with $\gforest{\delta}$, and thus,
	\[
	\chase{D}{\dep}\ =\ \bigcup_{\alpha \in D} \gtree{\delta,\alpha}.
	\] 
	Hence, $|\chase{D}{\dep}| \leq |D|  \cdot (d+1) \cdot ||\dep||^{2 \cdot \arity{\dep} \cdot (d+1)}$, as needed.

%
The way that Proposition~\ref{pro:generic-bound} is stated gives the impression that the provided upper bound on $|\chase{D}{\dep}|$ is at least exponential w.r.t.~$D$, since the depth of a null may depend on $D$, which is of course undesirable. Interestingly, as we shall see in the next sections, the depth of a term is actually bounded by an integer that depends only on the set of TGDs, and thus, $|\chase{D}{\dep}|$ is always linear in $|D|$, for all the classes of TGDs considered in this work. In particular, we will see that, for each class of TGDs $\class{C} \in \{\class{SL},\class{L},\class{G}\}$, given a database $D$ and a set $\dep \in \class{C} \cap \class{CT}_D$,
$
\mathsf{maxdepth}(D,\dep) \leq \mathsf{d}_{\class{C}}(\dep),
$
where the function $\mathsf{d}_{\class{C}} : \class{C} \ra \mathbb{N}$ is defined as follows:
\begin{eqnarray*}
	\mathsf{d}_{\class{SL}}(\dep) &=& |\sch{\dep}| \cdot \arity{\dep}\\
	\mathsf{d}_{\class{L}}(\dep) &=& |\sch{\dep}| \cdot \arity{\dep}^{\arity{\dep}+1}\\
	\mathsf{d}_{\class{G}}(\dep) &=& |\sch{\dep}| \cdot \arity{\dep}^{2 \cdot \arity{\dep}+1} \cdot 2^{|\sch{\dep}| \cdot \arity{\dep}^{\arity{\dep}}}.
\end{eqnarray*}

\OMIT{
\medskip

It is clear that, if $\chase{D}{\dep}$ is finite, then there exists an integer $k_{D,\dep} \geq 0$ such that, for each $t \in \adom{\chase{D}{\dep}}$, $\depth{t} \leq k_{D,\dep}$.
We are now ready to provide the generic bound on the number of atoms occurring in $\chase{D}{\dep}$ claimed above.
Let $\mathsf{chsize}(\cdot,\cdot)$ be the function from pairs consisting of a database and a set of TGDs to the natural numbers defined as expected, that is, $\mathsf{chsize}(D,\dep) = |\chase{D}{\dep}|$.

\begin{proposition}\label{pro:generic-bound}
	Consider a database $D$, and a set $\dep \in \class{G}$ of TGDs. If $\dep \in \class{CT}_{D}$, then
	\[
	\chasesize{D}{\dep}\ \in\ O\left(|D|  \cdot ||\dep||^{2 \cdot \arity{\dep} \cdot (d+1)} \cdot d\right),
	\]
	where $d = \max_{t \in \adom{\chase{D}{\dep}}} \{\depth{t}\}$.
\end{proposition}

\begin{proposition}\label{pro:generic-bound}
	Consider a database $D$, and a set $\dep \in \class{G}$ of TGDs. If $\dep \in \class{CT}_{D}$, then
	\[
	\chasesize{D}{\dep}\ \in\ |D|  \cdot ||\dep||^{O(\arity{\dep} \cdot d)} \cdot (d+1),
	\]
	where $d = \max_{t \in \adom{\chase{D}{\dep}}} \{\depth{t}\}$.
\end{proposition}

\begin{proposition}\label{pro:generic-bound}
	Consider a database $D$, and a set $\dep \in \class{G}$ of TGDs. If $\dep \in \class{CT}_{D}$, then
	\[
	\chasesize{D}{\dep}\ \in\ |D|  \cdot ||\dep||^{O(d)} \cdot \arity{\dep}^{O(\arity{\dep} \cdot d)} \cdot (d+1),
	\]
	where $d = \max_{t \in \adom{\chase{D}{\dep}}} \{\depth{t}\}$.
\end{proposition}

Note that the above result, whose proof can be found in the appendix, heavily relies on guardedness, and it does not hold for arbitrary sets of TGDs.
The way that Proposition~\ref{pro:generic-bound} is stated, it gives the impression that the number of atoms in $\chase{D}{\dep}$, whenever $\dep \in \class{G} \cap \class{CT}_D$, is at least exponential in the database $D$, since the depth of a null may depend on $D$, which is of course an undesirable outcome. Interestingly, as we shall see in the next sections, due to guardedness, the depth of a term depends only on the set of TGDs, which in turn implies that $|\chase{D}{\dep}|$ is always polynomial in $|D|$.
}
\section{Simple Linear TGDs}\label{sec:simple-linear}

We now concentrate on the class of simple linear TGDs, the easier class to analyse among the classes of TGDs considered in this work, and provide a characterization of non-uniform chase termination as the one described in Section~\ref{sec:chase-procedure}, together with a matching lower bound for the size of the chase instance. We then exploit this characterization to pinpoint the complexity of $\mathsf{ChTrm}(\class{SL})$.

We know from~\cite{CaGP15} that in the case of simple linear TGDs, uniform chase termination can be characterized via weak-acyclicity. Recall that weak-acyclicity was introduced in~\cite{FKMP05} as the main language for data exchange purposes, which guarantees the finiteness of the chase for {\em every} input database. 
Unsurprisingly, we employ a non-uniform version of weak-acyclicity, that is, a database-dependent variant of weak-acyclicity, for characterizing non-uniform chase termination. Let us clarify, however, that since we want to provide a worst-case optimal upper bound on the size of the chase instance, we need to rely on a more refined analysis than the one performed in~\cite{CaGP15}; the size of the chase instance was not considered in~\cite{CaGP15}.

\medskip

\noindent
\textbf{Non-Uniform Weak-Acyclicity.}
We first need to recall the notion of the {\em dependency graph} of a set $\dep$ of TGDs, 
defined as a directed multigraph $\depg{\dep}=(N,E)$, where $N = \pos{\sch{\dep}}$, and $E$ contains {\em only} the following edges.
For each TGD $\sigma \in \dep$ with $\head{\sigma} = \{\alpha_1,\ldots,\alpha_k\}$, for each $x \in \frontier{\sigma}$, and for each position $\pi \in \posvar{\body{\sigma}}{x}$:
\begin{itemize}
		\item[-] For each $i \in [k]$, and for each $\pi' \in \posvar{\alpha_i}{x}$, there exists a \emph{normal} edge $(\pi,\pi') \in E$.
		\item[-] For each existentially quantified variable $z$ in $\sigma$, $i \in [k]$, and $\pi' \in \posvar{\alpha_i}{z}$, there is a \emph{special} edge $(\pi,\pi') \in E$.
\end{itemize}

We further need to define when a predicate is reachable from another predicate. 
Given predicates $R,P \in \sch{\dep}$, we write $R \ra_\dep P$  if $R = P$, or there exists a TGD $\sigma \in \dep$ such that $R$ occurs in $\body{\sigma}$ and $P$ occurs in $\head{\sigma}$. We say that {\em $P$ is reachable from $R$ (w.r.t.~$\dep$)}, denoted $R \reach{\dep} P$, if (i) $R \ra_\dep P$, or (ii) there exists $T \in \sch{\dep}$ such that $R \reach{\dep} T$ and $T \ra_\dep P$.
Given a database $D$, we say that a (not necessarily simple, and possibly cyclic) path $C$ in $\depg{\dep}$ is \emph{$D$-supported} if there is $R(\bar t) \in D$ and a node $(P,i)$ in $C$ with $R \reach{\dep} P$.
%

\begin{definition}\label{def:dwa}
	Consider a database $D$, and a set $\dep$ of TGDs. We say that $\dep$ is {\em weakly-acyclic w.r.t.~$D$}, or simply {\em $D$-weakly-acyclic}, if there is no $D$-supported cycle in $\depg{\dep}$ with a special edge. \hfill\markfull
\end{definition}

The key properties of non-uniform weak-acyclicity that are crucial for our analysis are summarized by the following two technical lemmas; the proofs are in the appendix.
The first one, which holds for arbitrary TGDs, provides a database-independent upper bound on the depth of terms occurring in $\chase{D}{\dep}$ via $\mathsf{d}_{\class{SL}} : \class{SL} \ra \mathbb{N}$ defined in Section~\ref{sec:generic-bound}; recall that $\mathsf{d}_{\class{SL}}(\dep) = |\sch{\dep}| \cdot \arity{\dep}$.

\begin{lemma}\label{lem:depth-bound-sl}
	Consider a database $D$, and a set $\dep$ of TGDs that is $D$-weakly-acyclic. It holds that $\mathsf{maxdepth}(D,\dep) \leq \mathsf{d}_{\class{SL}}(\dep)$.
\end{lemma}

What is more interesting is the fact that whenever there exists a bound on the depth of the terms occurring in a chase instance $\chase{D}{\dep}$, for a database $D$ and set $\dep \in \class{SL}$, then $\dep$ is necessarily $D$-weakly-acyclic. For this property, simple linearity is crucial.

\begin{lemma}\label{lem:bounded-depth-implies-wa}
	Consider a database $D$, and a set $\dep \in \class{SL}$.
	If there is $k \geq 0$ such that $\mathsf{maxdepth}(D,\dep) \leq k$, then $\dep$ is $D$-weakly-acyclic.
\end{lemma}

As discussed above, (uniform) weak-acyclicity has been used in~\cite{CaGP15} to characterize uniform chase termination in the case of simple linear TGDs. However,~\cite{CaGP15} did not establish properties analogous to Lemmas~\ref{lem:depth-bound-sl} and~\ref{lem:bounded-depth-implies-wa} for (uniform) weak-acyclicity since estimating the size of the chase was not part of the investigation.

\medskip

\noindent
\textbf{Characterizing Non-Uniform Termination.}
We are now ready to establish the desired characterization of non-uniform chase termination in the case of simple linear TGDs. Let $f_\class{SL}$ be the function from $\class{SL}$ to the natural numbers $\mathbb{N}$ defined as follows:
\[
f_\class{SL}(\dep)\ =\ \left(\mathsf{d}_{\class{SL}}(\dep)+1\right) \cdot ||\dep||^{2 \cdot \arity{\dep} \cdot (\mathsf{d}_{\class{SL}}(\dep)+1)}.
\]
We can then show the following result:

\begin{theorem}\label{thm:characterization-simple-linear}
	Consider a database $D$, and a set $\dep \in \class{SL}$ of TGDs. The following are equivalent:
	\begin{enumerate}
		\item $\dep \in \class{CT}_D$.
		\item $|\chase{D}{\dep}|\ \leq\ |D|  \cdot f_\class{SL}(\dep)$.
		\item $\dep$ is $D$-weakly-acyclic.
	\end{enumerate}
\end{theorem}

\begin{proof}
The interesting directions are $(1) \Rightarrow (3)$ and $(3) \Rightarrow (2)$; $(2) \Rightarrow (1)$ is trivial.
For $(1) \Rightarrow (3)$, observe that if $\dep \in \class{CT}_D$, which means that $\chase{D}{\dep}$ is finite, then there is $k \geq 0$ that bounds $\mathsf{maxdepth}(D,\dep)$, and $(3)$ follows by Lemma~\ref{lem:bounded-depth-implies-wa}.
The statement $(3) \Rightarrow (2)$ follows from Proposition~\ref{pro:generic-bound} and Lemma~\ref{lem:depth-bound-sl}.
\end{proof}

We complement the above characterization with a result that essentially states the following: the size of the chase instance is unavoidably exponential in the arity and the number of predicates of the underlying schema; the proof can be found in the appendix.

\begin{theorem}\label{the:lower-bound-sl}
	There exists a family of databases $\{D_{\ell}\}_{\ell > 0}$ with $\ell = |D_{\ell}|$, and a family of sets of TGDs $\{\dep_{n,m} \in \class{SL} \cap \class{CT}_{D_\ell}\}_{n,m>0}$ with $n = |\sch{\dep_{n,m}}|-1$ and $m = \arity{\dep_{n,m}}$, such that
	\[
	|\chase{D_{\ell}}{\dep_{n,m}}|\ \geq\ \ell \cdot m^{n \cdot m}.
	\]
\end{theorem}


\noindent
\textbf{Complexity.} We now study the complexity of $\mathsf{ChTrm}(\class{SL})$. Observe that the naive approach, which relies on item (2) of Theorem~\ref{thm:characterization-simple-linear},  that explicitly constructs the chase instance, shows that $\mathsf{ChTrm}(\class{SL})$ is in \textsc{ExpTime} (even if we bound the arity, we still get \textsc{ExpTime}), and in \textsc{PTime} in data complexity, i.e., when the set of TGDs is considered fixed. 
Notice also that Theorem~\ref{the:lower-bound-sl} tells us that this is the best that we can achieve via the naive approach since there is no way to lower the upper bound on the size of the chase instance provided by Theorem~\ref{thm:characterization-simple-linear}. 
It turns out that the exact complexity of $\mathsf{ChTrm}(\class{SL})$ is significantly lower than what we get via the the naive approach.

\begin{theorem}\label{the:complexity-sl}
	$\mathsf{ChTrm}(\class{SL})$ is \textsc{NL}-complete, even for schemas with unary and binary predicates, and in \textsc{AC}$_0$ in data complexity.
\end{theorem}

The finer procedures that lead to Theorem~\ref{the:complexity-sl} rely on item (3) of Theorem~\ref{thm:characterization-simple-linear}. In particular, we know that $\dep \in \class{CT}_D$ iff $\dep$ is $D$-weakly-acyclic.
Thus, it suffices to show that, given a database $D$ and a set $\dep \in \class{SL}$ of TGDs, the problem of deciding whether $\dep$ is {\em not} $D$-weakly-acyclic is in \textsc{NL} in general, and in \textsc{AC}$_0$ when $\dep$ is fixed.
The former is shown via a nondeterministic algorithm that performs a reachability check on the dependency graph of $\dep$, while the latter is shown by constructing a union of conjunctive queries $Q_\dep$ (which depends only on $\dep$) such that $\dep$ is not $D$-weakly-acyclic iff $D$ satisfies $Q_\dep$.
The details can be found in the appendix.

The  \textsc{NL}-hardness of $\mathsf{ChTrm}(\class{SL})$, even for schemas with unary and binary predicates, is inherited from~\cite{CaGP15}, where uniform chase termination is studied. In particular,~\cite{CaGP15} shows that, given a set $\dep \in \class{SL}$ such that $\sch{\dep}$ consists of unary and binary predicates, the following problem is \textsc{NL}-hard: with $D_\dep = \{P(c) \mid P/1 \in \sch{\dep}\} \cup \{R(c,c) \mid R/2 \in \sch{\dep}\}$, is it the case that $\chase{D_\dep}{\dep}$ is finite?
\section{Linear TGDs}\label{sec:linear}

We now concentrate on the class of linear TGDs, and provide a characterization of non-uniform chase termination as the one described in Section~\ref{sec:chase-procedure}, together with a matching lower bound for the size of the chase instance. We then exploit this characterization to pinpoint the complexity of $\mathsf{ChTrm}(\class{L})$.
In contrast to simple linear TGDs, non-uniform weak-acyclicity is not powerful enough for characterizing the finiteness of the chase instance. 

\begin{example}
	Consider the database $D = \{R(a,b)\}$, and the singleton set $\dep$ consisting of the (non-simple) linear TGD
	\[
	R(x,x)\ \ra\ \exists z \, R(z,x). 
	\]
	It is easy to see that there is no trigger for $\dep$ on $D$. This means that $\chase{D}{\dep} = D$ is finite, whereas $\dep$ is {\em not} $D$-weakly-acyclic. \hfill\markfull
\end{example}

In~\cite{CaGP15}, an extended version of weak-acyclicity, called {\em critical-weak-acyclicity}, has been proposed for characterizing uniform chase termination in the case of linear TGDs. Therefore, one could employ a non-uniform version of critical-weak-acyclicity for characterizing non-uniform chase termination. However, due to the involved definition of critical-weak-acyclicity, it is not clear how a database-independent bound on the depth of the terms occurring in a chase instance can be established.
Hence, to obtain a result analogous to Theorem~\ref{thm:characterization-simple-linear}, we exploit a technique, called {\em simplification}, that converts linear TGDs into simple linear TGDs.
This is a rather folklore technique in the context of ontological query answering, but this is the first time that it is being applied in the context of chase termination. Thus, our main technical task will be to show that simplification preserves the finiteness of the chase, and, more importantly, the depth of the terms occurring in a chase instance. 
This is a non-trivial task since the chase instance of the simplified version of a set $\dep$ of linear TGDs is structurally different than the chase instance of $\dep$ in the sense that there is no immediate correspondence between their terms and atoms.

\medskip

\noindent
\textbf{Simplification.}
Let $\bar t = (t_1,\ldots,t_n)$ be a tuple of (not necessarily distinct) terms. We write $\unique{\bar t}$ for the tuple obtained from $\bar t$ by keeping only the first occurrence of each term in $\bar t$.
For example, if $\bar t = (x,y,x,z,y)$, then $\unique{\bar t} = (x,y,z)$.
For each $i \in [n]$, the \emph{identifier of $t_i$ in $\bar t$}, denoted $\id{\bar t}{t_i}$, is the integer that identifies the position of $\unique{\bar t}$ at which $t_i$ appears. 
We write $\id{}{\bar t}$ for the tuple $(\id{\bar t}{t_1},\ldots,\id{\bar t}{t_n})$.
For example, if $\bar t = (x,y,x,z,y)$, then $\id{}{\bar t} = (1,2,1,3,2)$.
For an atom $\alpha = R(\bar t)$, the {\em simplification of $\alpha$}, denoted $\simple{\alpha}$, is the atom $R_{\id{}{\bar t}}(\unique{\bar t})$. We can naturally refer to the simplification of a set of atoms.
For a tuple of variables $\bar x = (x_1,\ldots,x_n)$, a \emph{specialization of $\bar x$} is a function $f$ from $\bar x$ to $\bar x$ such that $f(x_1) = x_1$, and $f(x_i) \in \{f(x_1),\ldots,f(x_{i-1}),x_i\}$, for each $i \in \{2,\ldots,n\}$.
We write $f(\bar x)$ for $(f(x_1),\ldots,f(x_n))$. We are now ready to define the simplification of linear TGDs.

\begin{definition}\label{def:simplification}
	Consider a linear TGD $\sigma$ of the form
	\[
	R(\bar x) \ra \exists \bar z\, \psi(\bar y,\bar z), 
	\]
	where $\bar y \subseteq \bar x$, and a specialization $f$ of $\bar x$. The {\em simplification of $\sigma$ induced by $f$} is the simple linear TGD
	\[
	\simple{R(f(\bar x))} \rightarrow \exists \bar z\, \simple{\psi(f(\bar y),\bar z)}.
	\]
	We write $\simple{\sigma}$ for the set of all simplifications of $\sigma$ induced by some specialization of $\bar x$.
	For a set $\dep \in \class{L}$ of TGDs, the {\em simplification of $\dep$} is defined as the set
	\[
	\simple{\dep}\ =\ \bigcup_{\sigma \in \dep} \simple{\sigma}
	\]
	consisting only of simple linear TGDs. \hfill\markfull
\end{definition}

The next result, which is one of our main technical results, shows that the technique of simplification can be safely applied in the context of chase termination in the sense that it preserves the finiteness of the chase, as well as the maximal depth over all terms in a chase instance. The key ingredient underlying the proof of this result is a delicate correspondence between the terms and atoms in the instance $\chase{D}{\dep}$, for some database $D$ and set $\dep$ of linear TGD, with those in the instance $\chase{\simple{D}}{\simple{\dep}}$; the low-level details are deferred to the appendix.

\begin{proposition}\label{pro:simplification}
	Consider a database $D$, and a set $\dep \in \class{L}$. Then:\footnote{It is easy to see that (2) implies (1), but we explicitly state (1) since it is interesting in its own right, and it is also explicitly used in the proof of Theorem~\ref{thm:characterization-linear}.}
	\begin{enumerate}
		\item $\dep \in \class{CT}_{D}$ if and only if $\simple{\dep} \in \class{CT}_{\simple{D}}$.
		\item $\mathsf{maxdepth}(D,\dep) = \mathsf{maxdepth}(\simple{D},
		\simple{\dep})$.
	\end{enumerate}
\end{proposition}

\noindent
\textbf{Characterizing Non-Uniform Termination.}
We are now ready to establish a result for linear TGDs analogous to Theorem~\ref{thm:characterization-simple-linear} via simplification. To this end, we first provide a database-independent upper bound on the depth of terms occurring in a chase instance via $\mathsf{d}_{\class{L}} : \class{L} \ra \mathbb{N}$; recall that $\mathsf{d}_{\class{L}}(\dep) = |\sch{\dep}| \cdot \arity{\dep}^{\arity{\dep}+1}$.

\begin{lemma}\label{lem:depth-bound-linear}
	Consider a database $D$, and a set $\dep \in \class{L}$ with $\simple{\dep}$ being $\simple{D}$-weakly-acyclic. Then $\mathsf{maxdepth}(D,\dep) \leq \mathsf{d}_{\class{L}}(\dep)$.
\end{lemma}

\begin{proof}
By Lemma~\ref{lem:depth-bound-sl}, and item (2) of Proposition~\ref{pro:simplification},
\begin{eqnarray*}
\mathsf{maxdepth}(D,\dep) &\leq& \mathsf{d}_{\class{SL}}(\simple{\dep})\\
&=& |\sch{\simple{\dep}}| \cdot \arity{\simple{\dep}}.
\end{eqnarray*}
It is easy to see that $|\sch{\simple{\dep}}| \leq |\sch{\dep}| \cdot \arity{\dep}^{\arity{\dep}}$, and that $\arity{\simple{\dep}} \leq \arity{\dep}$. This implies that
\[
\mathsf{maxdepth}(D,\dep)\ \leq\ |\sch{\dep}| \cdot \arity{\dep}^{\arity{\dep}+1}
\]
and the claim follows by the definition of the function $\mathsf{d}_\class{L}$.
\end{proof}

Let $f_\class{L}$ be the function from $\class{L}$ to $\mathbb{N}$ defined as follows:
\[
f_\class{L}(\dep)\ =\ \left(\mathsf{d}_{\class{L}}(\dep)+1\right) \cdot ||\dep||^{2 \cdot \arity{\dep} \cdot (\mathsf{d}_{\class{L}}(\dep)+1)}.
\]
By exploiting the generic bound provided by Proposition~\ref{pro:generic-bound}, the direction $(1) \Rightarrow (3)$ of Theorem~\ref{thm:characterization-simple-linear}, item (1) of Proposition~\ref{pro:simplification}, and Lemma~\ref{lem:depth-bound-linear}, it is an easy task to show the following:

\begin{theorem}\label{thm:characterization-linear}
	Consider a database $D$, and a set $\dep \in \class{L}$ of TGDs. The following are equivalent:
	\begin{enumerate}
		\item $\dep \in \class{CT}_D$.
		\item $|\chase{D}{\dep}|\ \leq\ |D|  \cdot f_\class{L}(\dep)$.
		\item $\simple{\dep}$ is $\simple{D}$-weakly-acyclic.
	\end{enumerate}
\end{theorem}

\OMIT{
\begin{proof}
	It is clear that $(2) \Rightarrow (1)$ holds trivially.
	Assume now that $(1)$ holds, i.e., $\dep \in \class{CT}_D$. By item (1) of Proposition~\ref{pro:simplification}, we get that $\simple{\dep} \in \class{CT}_{\simple{D}}$. Since, by construction, $\simple{\dep} \in \class{SL}$, Theorem~\ref{thm:characterization-simple-linear} implies that $\simple{\dep}$ is $\simple{D}$-weakly-acyclic, i.e., $(3)$ holds, as needed.
	Finally, assume that $(3)$ holds, i.e., $\simple{\dep}$ is $\simple{D}$-weakly-acyclic. By Lemma~\ref{lem:depth-bound-linear}, $\mathsf{maxdepth}(D,\dep) \leq \mathsf{d}_{\class{L}}(\dep)$. Thus, by Proposition~\ref{pro:generic-bound}, we get that (2) holds, as needed.
\end{proof}
}

We complement the above characterization with a result that states the following: the size of the chase instance is unavoidably double-exponential in the arity, and exponential in the number of predicates of the underlying schema; the proof is in the appendix.

\begin{theorem}\label{the:lower-bound-linear}
	There exists a family of databases $\{D_{\ell}\}_{\ell > 0}$ with $\ell = |D_{\ell}|$, and a family of sets of TGDs $\{\dep_{n,m} \in \class{L} \cap \class{CT}_{D_\ell}\}_{n,m>0}$ with $n = |\sch{\dep_{n,m}}|-1$ and $m = \arity{\dep_{n,m}}-3$, such that
	\[
	|\chase{D_{\ell}}{\dep_{n,m}}|\ \geq\ \ell \cdot 2^{n \cdot \left(2^{m}-1\right)}.
	\]
\end{theorem}

\noindent
\textbf{Complexity.}
We proceed with the complexity of $\mathsf{ChTrm}(\class{L})$. We first note that the naive procedure obtained from item $(2)$ of Theorem~\ref{thm:characterization-linear} shows that $\mathsf{ChTrm}(\class{L})$ is in \textsc{2ExpTime}, in \textsc{ExpTime} if we bound the arity, and in \textsc{PTime} in data complexity. This is provably the best that we can get from the naive approach according to Theorem~\ref{the:lower-bound-linear}.
However, the exact complexity of $\mathsf{ChTrm}(\class{L})$ is lower.

\begin{theorem}\label{the:complexity-linear}
	$\mathsf{CT}(\class{L})$ is \textsc{PSpace}-complete, \textsc{NL}-complete for schemas of bounded arity, and in \textsc{AC}$_0$ in data complexity.
\end{theorem}

The finer procedures that lead to Theorem~\ref{the:complexity-linear} rely on item (3) of Theorem~\ref{thm:characterization-linear}. In particular, we know that $\dep \in \class{CT}_D$ iff $\simple{\dep}$ is $\simple{D}$-weakly-acyclic.
Thus, it suffices to show that, given a database $D$ and a set $\dep \in \class{L}$ of TGDs, the problem of deciding whether $\simple{\dep}$ is {\em not} $\simple{D}$-weakly-acyclic is in \textsc{PSpace} in general, in \textsc{NL} in the case of bounded arity, and in \textsc{AC}$_0$ when $\dep$ is fixed.
The \textsc{PSpace} and \textsc{NL} upper bounds are shown via a nondeterministic algorithm that performs a reachability check on the dependency graph of $\simple{\dep}$, but without explicitly constructing it. The \textsc{AC}$_0$ upper bound is shown by constructing a union of conjunctive queries $Q_\dep$ (which depends only on $\dep$) such that $\simple{\dep}$ is not $\simple{D}$-weakly-acyclic iff $D$ satisfies $Q_\dep$.
%
%
The hardness results are inherited from~\cite{CaGP15}, where it is shown that the problem is hard for the database consisting of all the atoms that can be formed using one constant and the predicates of the underlying schema.

\section{Guarded TGDs}\label{sec:guarded}

In this final section, we focus on the class of guarded TGDs, and provide a characterization of non-uniform chase termination as the one described in Section~\ref{sec:chase-procedure}, together with a matching lower bound for the size of the chase instance. 
We then exploit this characterization to pinpoint the complexity of $\mathsf{ChTrm}(\class{G})$.
Note that the work~\cite{CaGP15}, where the uniform termination problem for guarded TGDs has been studied, does not provide any syntactic characterization, which can then be used to devise a decision procedure, but it rather relies on a sophisticated alternating algorithm.
Towards our characterization, we combine a technique known as {\em linearization}~\cite{GoMP14}, which converts a set of guarded TGDs into a set of linear TGDs, with simplification used in the previous section.
Linearization is a very useful technique that has found several applications in the context of query answering~\cite{AmBe18,BDFLP20,BBMT17,GoMP14,GoMP20}. This is the first time, however, that it is used in the context of chase termination. Hence, as for simplification, our main technical task will be to show that linearization preserves the finiteness of the chase, and, more importantly, the depth of the terms occurring in a chase instance. The main challenge is the fact that the chase instance of the linearized version of a set $\dep$ of guarded TGDs is structurally very different than the chase instance of $\dep$ in the sense that there is no immediate correspondence between their terms and atoms.

\medskip

\noindent
\textbf{Linearization.} 
Due to space constraints, we only give a high-level description of linearization, while the details are deferred to the appendix.
Given a database $D$, a set $\dep \in \class{G}$ of TGDs, and an atom $\alpha \in \chase{D}{\dep}$, the {\em type of $\alpha$ (w.r.t.~$D$ and $\dep$)} is the set of atoms in $\chase{D}{\dep}$ that mention only terms from $\alpha$~\cite{CaGK13}. 
The key property of the type states that the set of atoms in $\chase{D}{\dep}$ that can be derived from $\alpha$ (used as a guard) is determined by the type of $\alpha$.

Now, for linearizing the database $D$, we encode each atom $\alpha \in D$ and its type as an atom of the form $[\tau](\cdot)$, where $\tau$ is a symbolic representation of $\alpha$ and its type. For example, if  $\alpha = R(a,b)$ and $\{R(a,b),S(b,a),T(a)\}$ is its type, we encode $R(a,b)$ and its type as
$
[R(1,2),\{S(2,1),T(1)\}](a,b).
$
The obtained database is denoted $\mathsf{lin}(D)$.
For the linearization of $\dep$, the intention is, for a TGD $\sigma \in \dep$, to encode the shape of the type $\tau$ of the guard of $\sigma$ in a predicate $[\tau]$, and then replace $\sigma$ with a linear TGD that uses in its body an atom of the form $[\tau](\cdot)$. We need an effective way, though, to compute the type of an atom $\alpha$ by completing its known part, which is inherited from the type of the guard atom that generates $\alpha$, with atoms that mention the new null values invented in $\alpha$. This exploits the main property of the type.
The obtained set of linear TGDs is $\mathsf{lin}(\dep)$.

The next result, which is another key technical result of this work, shows that the technique of linearization can be safely applied in the context of chase termination in the sense that it preserves the finiteness of the chase, as well as the maximal depth over all terms in a chase instance. It actually shows a property analogous to the one shown by Proposition~\ref{pro:simplification} for simplification, and the key ingredient of its proof is an intricate correspondence between the terms and atoms in the instance $\chase{D}{\dep}$, for some database $D$ and set $\dep$ of guarded TGD, with those in the instance $\chase{\mathsf{lin}(D)}{\mathsf{lin}(\dep)}$; the low-level details are deferred to the appendix.

\begin{proposition}\label{pro:linearization}
	Consider a database $D$, and a set $\dep \in \class{G}$. Then:\footnote{It is easy to see that (2) implies (1), but we explicitly state (1) since it is interesting in its own right, and it is also used in the proof of Theorem~\ref{thm:characterization-guarded}.}
	\begin{enumerate}
		\item $\dep \in \class{CT}_{D}$ if and only if $\mathsf{lin}(\dep) \in \class{CT}_{\mathsf{lin}(D)}$.
		\item $\mathsf{maxdepth}(D,\dep) = \mathsf{maxdepth}(\mathsf{lin}(D),\mathsf{lin}(\dep))$.
	\end{enumerate}
\end{proposition}

\noindent
\textbf{Characterizing Non-Uniform Termination.} We proceed to establish a result for guarded TGDs analogous to Theorems~\ref{thm:characterization-simple-linear} and~\ref{thm:characterization-linear} via a combination of linearization and simplification. To this end, we first provide a database-independent upper bound on the depth of terms occurring in a chase instance via the function $\mathsf{d}_{\class{G}} : \class{G} \ra \mathbb{N}$ introduced in Section~\ref{sec:generic-bound}; recall that $\mathsf{d}_{\class{G}}(\dep) = |\sch{\dep}| \cdot \arity{\dep}^{2 \cdot \arity{\dep}+1} \cdot 2^{|\sch{\dep}| \cdot \arity{\dep}^{\arity{\dep}}}$. In what follows, for brevity, we write $\gsimple{\cdot}$ instead of $\simple{\mathsf{lin}(\cdot)}$.

\begin{lemma}\label{lem:depth-bound-guarded}
	Consider a database $D$, and $\dep \in \class{G}$ with $\gsimple{\dep}$ being $\gsimple{D}$-weakly-acyclic. Then $\mathsf{maxdepth}(D,\dep) \leq \mathsf{d}_{\class{G}}(\dep)$.
\end{lemma}

\begin{proof}
	By Lemma~\ref{lem:depth-bound-sl}, and item (2) of Propositions~\ref{pro:simplification} and~\ref{pro:linearization},
	\begin{eqnarray*}
		\mathsf{maxdepth}(D,\dep) &\leq& \mathsf{d}_{\class{SL}}(\gsimple{\dep})\\
		&=& |\sch{\gsimple{\dep}}| \cdot \arity{\gsimple{\dep}}.
	\end{eqnarray*}
	It is not difficult to show that
	\[
		|\sch{\gsimple{\dep}}|\ \leq\ |\sch{\mathsf{lin}(\dep)}| \cdot \arity{\mathsf{lin}(\dep)}^{\arity{\mathsf{lin}(\dep)}}.
	\]
	We further know that
	\[
	|\sch{\ling{\dep}}|\ \leq\ |\sch{\dep}| \cdot \arity{\dep}^{\arity{\dep}} \cdot 2^{|\sch{\dep}| \cdot \arity{\dep}^{\arity{\dep}}}.
	\]
	Since $\arity{\gsimple{\dep}} \leq \arity{\mathsf{lin}(\dep)} = \arity{\dep}$, we can conclude that $\mathsf{maxdepth}(D,\dep) \leq \mathsf{d}_\class{G}(\dep)$, and the claim follows.
\end{proof}

Let $f_\class{G}$ be the function from $\class{G}$ to $\mathbb{N}$ defined as follows:
\[
f_\class{G}(\dep)\ =\ \left(\mathsf{d}_{\class{G}}(\dep)+1\right) \cdot ||\dep||^{2 \cdot \arity{\dep} \cdot (\mathsf{d}_{\class{G}}(\dep)+1)}.
\]
By using the generic bound provided by Proposition~\ref{pro:generic-bound}, the direction $(1) \Rightarrow (3)$ of Theorem~\ref{thm:characterization-linear}, item (1) of Proposition~\ref{pro:linearization}, and Lemma~\ref{lem:depth-bound-guarded}, it is an easy task to show the following:

\begin{theorem}\label{thm:characterization-guarded}
	Consider a database $D$, and a set $\dep \in \class{G}$ of TGDs. The following are equivalent:
	\begin{enumerate}
		\item $\dep \in \class{CT}_D$.
		\item $|\chase{D}{\dep}|\ \leq\ |D|  \cdot f_\class{G}(\dep)$.
		\item $\gsimple{\dep}$ is $\gsimple{D}$-weakly-acyclic.
	\end{enumerate}
\end{theorem}


We complement the above with a result that states the following: the size of the chase instance is unavoidably triple-exponential in the arity, and double-exponential in the number of predicates of the schema; the rather involved construction is in the appendix.

\begin{theorem}\label{the:lower-bound-guarded}
	There exists a family of databases $\{D_{\ell}\}_{\ell > 0}$ with $\ell = |D_{\ell}|$, and a family of sets of TGDs $\{\dep_{n,m} \in \class{G} \cap \class{CT}_{D_\ell}\}_{n,m>0}$ with $n = \frac{|\sch{\dep_{n,m}}|}{4}-3$ and $m = \frac{\arity{\dep_{n,m}}}{2}-2$, such that
	\[
	|\chase{D_{\ell}}{\dep_{n,m}}|\ \geq\ \ell \cdot 2^{\left(2^n \cdot \left(2^{\left(2^m\right)}-1\right)\right)}.
	\]
\end{theorem}


\noindent
\textbf{Complexity.} We proceed with the complexity of $\mathsf{ChTrm}(\class{G})$. We first note that the naive procedure obtained from item $(2)$ of Theorem~\ref{thm:characterization-guarded} shows that $\mathsf{ChTrm}(\class{G})$ is in \textsc{3ExpTime}, in \textsc{2ExpTime} if we bound the arity, and in \textsc{PTime} in data complexity.
This is provably the best that we can get from the naive approach according to Theorem~\ref{the:lower-bound-guarded}.
However, apart from the case of data complexity, the exact complexity of $\mathsf{ChTrm}(\class{G})$ is significantly lower.

\begin{theorem}\label{the:complexity-guarded}
	$\mathsf{ChTrm}(\class{G})$ is \textsc{2ExpTime}-complete, \textsc{ExpTime}-complete for bounded arity, and \textsc{PTime}-complete in data complexity.
\end{theorem}

The above result is shown by explicitly constructing $\gsimple{D}$ and $\gsimple{\dep}$, on input $D$ and $\dep$, and then relying on the procedure for $\mathsf{ChTrm}(\class{SL})$. 
%
The \textsc{2ExpTime}/\textsc{ExpTime} lower bounds are inherited from~\cite{CaGP15}, where uniform chase termination is studied. The \textsc{PTime}-hardness in data complexity is shown by exploiting a technique from~\cite{CaGP15}, known as {\em looping operator}, which allows us to transfer hardness results for ontological query answering to our problem.
\section{Conclusions}\label{sec:conclusions}

%
%
The results of this work provide a rather complete picture concerning the size and complexity of the non-uniformly terminating semi-oblivious chase for guarded TGDs, and subclasses thereof.
The fact that for linear TGDs the problem  can be solved by evaluating a UCQ, which depends only on the TGDs, over the given database is particularly interesting as it allows us to exploit standard database systems. We are currently working on an experimental evaluation of the obtained procedures.
An interesting research direction is to perform a similar analysis for the restricted version of the chase, for which we have indications that it will be even more challenging.


\bibliographystyle{ACM-Reference-Format}

\bibliography{references}


\newpage
\appendix

\section{Undecidability of $\mathsf{ChTrm}(\dep^\star)$}

We provide a reduction from the halting problem. In fact, we adapt the reduction given in~\cite{DeNR08} so that only the database depends on the Turing machine, while the set of TGDs is fixed.
Consider a deterministic Turing machine $M = (S,\Lambda,f,s_0)$, where $S$ is the state set, $\Lambda$ is the tape alphabet, $f: S \times \Lambda \ra S \times \Lambda \times \{\ra,-,\leftarrow\}$ is the transition function, and $s_0 \in S$ is the initial state. We assume that $\Lambda$ contains the special symbols $\triangleright$ (marking the beginning of the tape), $\triangleleft$ (marking the end of the tape), and $\sqcup$ (blank symbol).
We are going to define a database $D_M$ such that the following are equivalent:
\begin{enumerate}
	\item $M$ halts on the empty input.
	\item There exists a finite chase derivation of $D_M$ w.r.t.~$\dep^\star$ (i.e., $\chase{D_M}{\dep^\star}$ is finite), where $\dep^\star$ is a fixed set of TGDs.
\end{enumerate}

\subsection*{The Database $D_M$}

The database essentially stores the transition function of $M$, as well as the initial configuration of $M$ on the empty input. It also stores some auxiliary atoms that will allow us to access special constants, and thus, keep $\dep^\star$ constant-free.

We use the $5$-ary predicate ${\rm Trans}$ to store the transition rules of $M$; ${\rm Trans}(x_1,x_2,x_3,x_4,x_5)$ encodes the fact that $f(x_1,x_2) = (x_3,x_4,x_5)$.
We also use the ternary predicates ${\rm Tape}$ and ${\rm Head}$ to store the content of the tape, the position of the head, and the current state; ${\rm Tape}(x,z,y)$ means that the cell $(x,y)$, which should be seen as a ``horizontal'' edge from $x$ to $y$ in a grid, contains the symbol $z$, while  ${\rm Head}(x,z,y)$ means that the head points at cell $(x,y)$ and the state is $z$.
Formally, the database $D_M$ is defined as
\begin{eqnarray*}
&& \{{\rm Trans}(s,a,s',a',d) \mid (s,a) \in S \times \Lambda \text{ and } f(s,a) = (s',a',d)\}\\
&\cup& \{{\rm Tape}(c_0,\triangleright,c_1), {\rm Tape}(c_1,\sqcup,c_2), {\rm Head}(c_1,s_0,c_2),{\rm Tape}(c_2,\triangleleft,c_3)\}\\
&\cup& \{{\rm LDir}(\leftarrow),{\rm SDir}(-),{\rm RDir}(\ra)\}\\
&\cup& \{{\rm Blank}(\sqcup),{\rm End}(\triangleleft)\} \\
&\cup& \{{\rm NormSymb}(a) \mid a \in \Lambda \setminus \{\triangleright,\triangleleft\}\}.
\end{eqnarray*}
This complete the construction of $D_M$.

\subsection*{The Fixed Set $\dep^\star$ of TGDs}

The goal of $\dep^\star$ is to simulate the computation of $M$ on the empty input. This is done by encoding the computation of $M$ as a grid with ``horizontal'' and ``vertical'' edges. Each row of the  grid (i.e., a sequence of nodes connected via ``horizontal'' edges) encodes a configuration of $M$. Moreover, successive rows indicate successive configurations, which are connected via two kinds of ``vertical'' edges, those to the left and those to the right of the head. We proceed with the definition of $\dep^\star$.

We have TGDs that execute the transition rules of $M$ that move the head to the right. We need to distinguish the two cases where the head reaches the end of the tape or not:
\begin{multline*}
{\rm Trans}(x_1,x_2,x_3,x_4,x_5), {\rm RDir}(x_5), {\rm NormSymb}(w),\\
{\rm Head}(x,x_1,y), {\rm Tape}(x,x_2,y), {\rm Tape}(y,w,z)\ \ra\\
\exists x' \exists y' \exists z' \, L(x,x'),R(y,y'),R(z,z'),\\
{\rm Tape}(x',x_4,y'),{\rm Head}(y',x_3,z'),{\rm Tape}(y',w,z').
\end{multline*}
and
\begin{multline*}
{\rm Trans}(x_1,x_2,x_3,x_4,x_5), {\rm RDir}(x_5), {\rm Blank}(u),{\rm End}(w),\\
{\rm Head}(x,x_1,y),
{\rm Tape}(x,x_2,y), {\rm Tape}(y,w,z)\ \ra\\
\exists x' \exists y' \exists z' \exists w' \, L(x,x'),R(y,y'),R(z,z'),\\
{\rm Tape}(x',x_4,y'),{\rm Head}(y',x_3,z'),\\
{\rm Tape}(y',u,z'),{\rm Tape}(z',w,w').
\end{multline*}

We also have a TGD that executes the transition rules that move the head to the left. We assume, w.l.o.g., that $M$ is well-behaved, and never reads beyond the first cell:
\begin{multline*}
{\rm Trans}(x_1,x_2,x_3,x_4,x_5), {\rm LDir}(x_5),\\
{\rm Tape}(x,w,y),{\rm Head}(y,x_1,z), {\rm Tape}(y,x_2,z)\ \ra\\
\exists x' \exists y' \exists z' \, R(x,x'),R(y,y'),L(z,z'),\\
{\rm Head}(x',x_3,y'),{\rm Tape}(x',w,y'),{\rm Tape}(y',x_4,z').
\end{multline*}

We also have a TGD that executes the transition rules that do not move the head:
\begin{multline*}
{\rm Trans}(x_1,x_2,x_3,x_4,x_5), {\rm SDir}(x_5),\\
{\rm Head}(x,x_1,y), {\rm Tape}(x,x_2,y)\ \ra\\
\exists x' \exists y' \, L(x,x'),R(y,y'),\\
{\rm Head}(x',x_3,y'),{\rm Tape}(x',x_4,y').
\end{multline*} 

We finally need to copy the content of the cells that remained untouched after a transition. This is done by copying the content of the cells to the left and to the right of the head:
\begin{eqnarray*}
	{\rm Tape}(x,z,y), L(y,y') &\ra& \exists x' \, L(x,x'),{\rm Tape}(x',z,y')\\
	{\rm Tape}(x,z,y), R(x,x') &\ra& \exists y' \, {\rm Tape}(x',z,y'),R(y,y').
\end{eqnarray*}
This completes the definition of $\dep^\star$. It is clear that $\dep^\star$ does not depend on the Turing machine $M$.
It is not difficult to verify that $M$ halts on the empty input iff there exists a finite chase derivation of $D_M$ w.r.t.~$\dep$ (i.e., $\chase{D_M}{\dep}$ is finite), and the claim follows.


\section{Proof of Lemma~\ref{lem:depth-bound}}

We start by making a crucial observation due to guardedness.
Consider an arbitrary atom $\beta \in \gtree{\delta,\alpha}$, and let $T$ be the set of all terms occurring in $\beta$ that appear also in some ancestor of $\beta$ in $\gtree{\delta,\alpha}$. Then, there exists an ancestor $\gamma$ of $\beta$ in $\gtree{\delta,\alpha}$ that mentions all the terms of $T$. In other words, all the terms of $\beta$ that are not invented in $\beta$ (according to the derivation $\delta$), occur all together in a single ancestor of $\beta$ in $\gtree{\delta,\alpha}$. We are now ready to prove our claim by induction on the depth $i \geq 0$.

\medskip

\noindent
\textbf{Base Case:} Consider an atom $\beta \in \mathsf{gtree}^{0}(\delta,\alpha)$. From the above observation, we get that all the terms in $\beta$ occur also in $\alpha$. Thus, the maximum number of facts in $\gtree{\delta,\alpha}$ of depth 0 is the number of all possible atoms that can be formed using predicates of $\sch{\dep}$ and constants in $\alpha$.
Since $\alpha$ contains at most $\arity{\dep}$ distinct constants,
\begin{eqnarray*}
|\mathsf{gtree}^{0}(\delta,\alpha)| &\leq& |\sch{\dep}| \cdot \arity{\dep}^{\arity{\dep}}\\
&\leq& ||\dep||^{2 \cdot \arity{\dep}},
\end{eqnarray*}
and the claim follows.

\medskip

\noindent
\textbf{Induction Step:} Let $i>0$. We write $\mathsf{gtree}^{i}_{\exists}(\delta,\alpha)$ for the set of atoms of $\mathsf{gtree}^{i}(\delta,\alpha)$ in which a new null of depth $i$ is invented. From the observation made above, we conclude that $\beta \in \mathsf{gtree}^{i}(\delta,\alpha) \setminus \mathsf{gtree}^{i}_{\exists}(\delta,\alpha)$ implies $\beta$ has an ancestor $\gamma$ in $\mathsf{gtree}^{i}(\delta,\alpha)$ such that $\gamma \in \mathsf{gtree}^{i}_{\exists}(\delta,\alpha)$ and it contains all the terms occurring in $\beta$. Therefore, an atom $\beta$ occurs in $\mathsf{gtree}^{i}(\delta,\alpha)$ only if all its terms occur together in an atom of $\mathsf{gtree}^{i}_{\exists}(\delta,\alpha)$.
Since each atom of $\mathsf{gtree}^{i}_{\exists}(\delta,\alpha)$ mentions at most $\arity{\dep}$ distinct terms, we get that
\[
|\mathsf{gtree}^{i}(\delta,\alpha)|\ \leq\ |\mathsf{gtree}^{i}_{\exists}(\delta,\alpha)| \cdot |\sch{\dep}| \cdot \arity{\dep}^{\arity{\dep}}.
\]
We now provide an upper bound on $|\mathsf{gtree}^{i}_{\exists}(\delta,\alpha)|$. Consider an atom $\beta \in \mathsf{gtree}^{i}_{\exists}(\delta,\alpha)$, and let $(\sigma,h)$ be the trigger in $\delta$ that generates $\beta$. From our observation above, we get that all the terms $h(\fr{\sigma})$ occur in an ancestor $\gamma$ of $\beta$ in $\mathsf{gtree}(\delta,\alpha)$. Moreover, since $\depth{\beta} = i$, the maximum depth over all terms in $h(\fr{\sigma})$ is $i-1$, and thus, $\gamma \in \mathsf{gtree}^{i-1}(\delta,\alpha)$. Since $|\fr{\sigma}| \leq \arity{\dep}$, we get that
\[
|\mathsf{gtree}^{i}_{\exists}(\delta,\alpha)|\ \leq\ |\atoms{\dep}| \cdot |\mathsf{gtree}^{i-1}(\delta,\alpha)| \cdot \arity{\dep}^{\arity{\dep}},
\]
Hence, we have that
\begin{multline*}
|\mathsf{gtree}^{i}(\delta,\alpha)|\ \leq\ |\atoms{\dep}| \cdot |\mathsf{gtree}^{i-1}(\delta,\alpha)| \cdot\\ \arity{\dep}^{\arity{\dep}} \cdot |\sch{\dep}| \cdot \arity{\dep}^{\arity{\dep}}.
\end{multline*}
By induction hypothesis, we get that
\begin{multline*}
|\mathsf{gtree}^{i}(\delta,\alpha)|\ \leq\ |\atoms{\dep}| \cdot ||\dep||^{2 \cdot \arity{\dep} \cdot i} \cdot\\ \arity{\dep}^{\arity{\dep}} \cdot |\sch{\dep}| \cdot \arity{\dep}^{\arity{\dep}}.
\end{multline*}
Since $|\atoms{\dep}| \cdot |\sch{\dep}| \cdot \arity{\dep}^{2 \cdot \arity{\dep}} \leq ||\dep||^{2 \cdot \arity{\dep}}$,
\begin{eqnarray*}
|\mathsf{gtree}^{i}(\delta,\alpha)| &\leq& ||\dep||^{2 \cdot \arity{\dep} \cdot i} \cdot ||\dep||^{2 \cdot \arity{\dep}}\\
&=& ||\dep||^{2 \cdot \arity{\dep} \cdot (i+1)},
\end{eqnarray*}
and the claim follows.


\section{Proof of Lemma~\ref{lem:depth-bound-sl}}

We show that, for every $t \in \adom{\chase{D}{\dep}}$, $\depth{t} \leq \mathsf{d}_{\class{SL}}(\dep)$, which implies that $\mathsf{maxdepth}(D,\dep) \leq \mathsf{d}_{\class{SL}}(\dep)$.
For a node $(R,i)$ in $\depg{\dep}$, an incoming path to $(R,i)$ is a (finite or infinite) path ending at $(R,i)$. We define the rank of $(R,i)$ as the maximum number of special edges over all incoming paths to $(R,i)$. It is clear that the nodes of $\depg{\dep}$ can be partitioned into $\{\Pi_\infty,\Pi_F\}$, where $\Pi_\infty$ (resp., $\Pi_F$) collects all the nodes with infinite (resp., finite) rank.

We first argue that no term in $\chase{D}{\dep}$ occurs at a position of $\Pi_\infty$. Consider a position $(P,i) \in \Pi_\infty$. It is clear that $(P,i)$ occurs in a cycle $C$ in $\depg{\dep}$ with a special edge. Since $\dep$ is $D$-weakly-acyclic, we conclude that $C$ is {\em not} $D$-supported. The latter means that there is not atom $R(\bar t) \in D$ such that $R \reach{\dep} P$. Therefore, there is no term in $\chase{D}{\dep}$ occurring at $(P,i)$.

We now proceed to show that the terms in $\chase{D}{\dep}$ occurring at positions of $\Pi_F$ are of depth at most $\mathsf{d}_{\class{SL}}(\dep)$. To this end, we partition the positions of $\Pi_F$ into $\{\Pi_0,\ldots,\Pi_r\}$, where, for each $i \in \{0,\ldots,r\}$, $\Pi_r$ collects all the nodes in $\depg{\dep}$ with rank $i$. We proceed to show the following claim:

\begin{claim}\label{cla:bound-depth}
	For each $i \in \{0,\ldots,r\}$, and term $t$ in $\chase{D}{\dep}$ occurring at position of $\Pi_i$, it holds that $\depth{t} \leq i$.
\end{claim}

\begin{proof}
	We proceed by induction on $i \in \{0,\ldots,r\}$.

\medskip

\noindent
\textbf{Base Case:} For $i=0$, we immediately get that $t \in \adom{D}$, which in turn implies that $\depth{t} = 0$, as needed.

\medskip

\noindent
\textbf{Induction Step:} Assume that $t$ appears in $\chase{D}{\dep}$ at a position $\pi \in \Pi_i$. If $t \in \adom{D}$, then $\depth{t} = 0 \leq i$. Assume now that $t$ is a null value. Note that $t$ may occur in $\chase{D}{\dep}$ at position $\pi$ for the following two reasons:
\begin{enumerate}
	\item By being generated during a chase step as a new null.
	\item By being propagated via a sequence of chase steps from the position $\pi' \in \Pi_j$, for $j \in \{0,\ldots,r\}$, at which it has been generate during a chase step.
\end{enumerate}

In case (1), note hat $t$ is generated at position $\pi$ during a chase step due to a special edge. Clearly, such an incoming special edge to $\pi$ must start at a node of $\bigcup_{k=0}^{i-1} \Pi_k$; otherwise, the rank of $\pi$ is greater that $i$, which is a contradiction. Thus, by induction hypothesis, we get that $\depth{t} \leq i$.

In case (2), we have a path of (non-special) edges from $\pi' \in \Pi_j$ to $\pi$ in $\depg{\dep}$. Observe that such a path can originate only from positions of $\bigcup_{k=0}^{i} 
\Pi_k$, that is, $j \leq i$; otherwise, the rank of $\pi$ is greater that $i$, which is a contradiction. 
In case $j < i$, by induction hypothesis, $\depth{t} \leq j < i$.
Now, if $j=i$, then, by applying a similar argument as in case (1) for position $\pi'$, we can show that $\depth{t} \leq i$. This completes the proof of Claim~\ref{cla:bound-depth}.
\end{proof}

To conclude the proof of Lemma~\ref{lem:depth-bound-sl}, observe that 
\[
r\ \leq\ |\sch{\dep}| \cdot \arity{\dep}\ =\ \mathsf{d}_{\class{SL}}(\dep).
\] 
Hence, by Claim~\ref{cla:bound-depth}, for each term $t$ in $\chase{D}{\dep}$ occurring at a position of $\Pi_F$, it holds that $\depth{t} \leq \mathsf{d}_{\class{SL}}(\dep)$.


\section{Proof of Lemma~\ref{lem:bounded-depth-implies-wa}}

Assume that $\dep$ is not $D$-weakly-acyclic. We proceed to show that there is no $k \geq 0$ such that $\mathsf{maxdepth}(D,\dep) \leq k$.
By hypothesis, there exists a $D$-supported cycle $C$ in $\depg{\dep}$ that contains a special edge. We first make a useful observation.
Since, by definition, no variable occurs twice in the body of a simple linear TGD, the existence of an atom in $D$ with predicate $R$, and the fact that $R \reach{\dep} P$, imply that there exists an atom in $\chase{D}{\dep}$ with predicate $P$. Moreover, for every predicate $S$ that occurs in $C$, there exists an atom with predicate $S$ in $\chase{D}{\dep}$.
Assume now that $((P,i),(T,j))$ is one of the special edges in $C$, and let $P(t_1,\ldots, t_n)$ be an atom in $\chase{D}{\dep}$, which exists due to the above observation. By definition of the dependency graph, the existence of $((P,i),(T,j))$ implies that there exists a TGD $\sigma \in \dep$ such that the variable that occurs at position $(P,i)$ in $\body{\sigma}$ is a frontier variable, and there exists an atom $\alpha$ in $\head{\sigma}$ with predicate $T$ having at position $(T,j)$ an existentially quantified  variable. Hence, since $\sigma$ is simple linear (i.e., no variable occurs more than once in $\body{\dep}$), $\chase{D}{\dep}$ must contain an atom of the form $T(u_1,\ldots,u_m)$ such that $\depth{u_j} > \depth{t_i}$.
Moreover, due to the cycle $C$, there must be another atom of the form $P(t'_1,\ldots,t'_m)$, where $\depth{t'_i} \ge \depth{u_j} > \depth{t_i}$; the latter again relies on the fact that $\dep$ consists of simple linear TGDs.
By iteratively applying the above argument, we can conclude that, for every integer $k  \ge 0$, there exists an atom of the form $P(t_1,\ldots,t_n)$ in $\chase{D}{\dep}$, where $\depth{t_i} > k$, and the claim follows.


\section{Proof of Theorem~\ref{the:lower-bound-sl}}
%

For each $\ell > 0$, we define the database
\[
D_\ell\ =\ \{R_0(c_1),\ldots,R_0(c_\ell)\}.
\]
Now, for each $n,m>0$, we define the set of TGDs
\[
\dep_{n,m}\ =\ \dep_{\text{start}}\ \cup\ \bigcup_{i=1}^{n} \dep_{i}^{\forall}\ \cup\ \bigcup_{i=1}^{n-1} \dep_{i}^{\exists},
\]
where $\dep_{\text{start}}$ consists of the TGD
\[
R_0(x) \ra \exists y_1 \cdots \exists y_m \, R_0(x),R_1(y_1,\ldots,y_m), 
\]
$\dep_{i}^{\forall}$ contains, for each $j \in [m]$, the TGDs
\begin{multline*}
	R_i(x_1,x_2,\ldots,x_{j-1},x_j,x_{j+1},\ldots,x_m)\ \ra\\ R_i(x_j,x_2,\ldots,x_{j-1},x_1,x_{j+1},\ldots,x_m)
\end{multline*}
and
\[
	R_i(x_1,x_2,\ldots,x_j,\ldots,x_m)\ \ra\ R_i(x_j,x_2\ldots,x_j,\ldots,x_m),
\]
and $\dep_{i}^{\exists}$ consists of the TGD
\[
R_i(x_1,\ldots,x_m)\ \ra\
\exists z_1 \cdots \exists z_m \, R_i(x_1,\ldots,x_m),
R_{i+1}(z_1,\ldots,z_m).
\]
This completes the construction of $\dep_{n,m}$. It can be verified that, for each $\ell,n,m>0$, $\dep_{n,m} \in \class{SL} \cap \class{CT}_{D_{\ell}}$. It remains to show that $|\chase{D_{\ell}}{\dep_{n,m}}| \geq \ell \cdot m^{n \cdot m}$. To this end, we first establish the following auxiliary claim:

\begin{claim}\label{cla:lower-bound-sl}
	Let $D = \{R_0(c)\}$, where $c \in \ins{C}$. For each $i \in [n]$,
	\[
	|\{\bar t \mid R_i(\bar t) \in \chase{D}{\dep_{n,m}}\}|\ =\ m^{i \cdot m}.
	\]
\end{claim}

\begin{proof}
	We proceed by induction on $i \in [n]$.

\medskip

\noindent
\textbf{Base Case:} It is clear that atoms with predicate $R_1$ are only generated due to the TGD $\sigma_{\text{start}}$ of $\dep_{\text{start}}$, and the TGDs of $\dep_{1}^{\forall}$. In fact, the database atom $R_0(c)$ triggers $\sigma_{\text{start}}$, and we get an atom of the form $R_1(\bot_{1}^{c},\ldots,\bot_{m}^{c})$, where $\bot_{1}^{c},\ldots,\bot_{m}^{c}$ are distinct nulls. Then, it is an easy exercise to verify that, starting from the atom $R_1(\bot_{1}^{c},\ldots,\bot_{m}^{c})$, the TGDs of $\dep_{1}^{\forall}$ generate all the possible atoms $R_1(\bar t)$ with $\bar t \in \{\bot_{1}^{c},\ldots,\bot_{m}^{c}\}^{m}$. Therefore, $|\{\bar t \mid R_1(\bar t) \in \chase{D}{\dep_{n,m}}\}| = m^{m}$.

\medskip

\noindent
\textbf{Induction Step:} Due to the TGD of $\dep_{i-1}^{\exists}$, for $i > 1$, for each atom $R_{i-1}(\bar t) \in \chase{D}{\dep_{n,m}}$, we have an atom $R_i(\bot_{1}^{\bar t},\ldots,\bot_{m}^{\bar t})$, where $\bot_{1}^{\bar t},\ldots,\bot_{m}^{\bar t}$ are distinct nulls. Furthermore, starting from the atom $R_i(\bot_{1}^{\bar t},\ldots,\bot_{m}^{\bar t})$, the TGDs of $\dep_{i}^{\forall}$ generate all the possible atoms $R_i(\bar u)$ with $\bar u \in \{\bot_{1}^{\bar t},\ldots,\bot_{m}^{\bar t}\}^{m}$. Therefore, for each atom $R_{i-1}(\bar t) \in \chase{D}{\dep_{n,m}}$ we have $m^m$ atoms in $\chase{D}{\dep_{n,m}}$ with predicate $R_i$. By induction hypothesis, in $\chase{D}{\dep_{n,m}}$ we have $m^{(i-1)\cdot m}$ atoms with predicate $R_{i-1}$, and thus,
\[
|\{\bar t \mid R_i(\bar t) \in \chase{D}{\dep_{n,m}}\}|\ =\ m^{(i-1) \cdot m} \cdot m^m\ =\ m^{i \cdot m},
\]
and the claim follows.
\end{proof}

By Claim~\ref{cla:lower-bound-sl}, it is straightforward to see that
\[
|\{\bar t \mid R_n(\bar t) \in \chase{D_\ell}{\dep_{n,m}}\}|\ =\ \ell \cdot m^{n \cdot m},
\]
which in turn implies that
\[
|\chase{D_{\ell}}{\dep_{n,m}}|\ \geq\ \ell \cdot m^{n \cdot m}.
\]


\section*{Proof of Theorem~\ref{the:complexity-sl}}

Consider a database $D$, and a set $\dep \in \class{SL}$ of TGDs. By Theorem~\ref{thm:characterization-simple-linear}, we know that $\dep \in \class{CT}_D$ iff $\dep$ is $D$-weakly-acyclic. 
Since co\textsc{NL} = \textsc{NL}, for the \textsc{NL} upper bound it suffices to show the following lemma:

\begin{algorithm}[t]
	\KwIn{A database $D$, and a set $\dep \in \class{SL}$ of TGDs}
	\KwOut{$\mathsf{accept}$ if $\dep$ is not $D$-weakly-acyclic; otherwise, $\mathsf{reject}$}
	\vspace{2mm}
	
	$\mathit{flag} := 0$\\
	\If{\text{\rm there is no edge} $(\pi,\pi')$ \text{\rm in} $\depg{\dep}$}{$\mathsf{reject}$}
	\textbf{guess} \text{an edge} $(\pi,\pi')$ \text{in} $\depg{\dep}$ \text{with} $\pi = (P,i)$\\
	\If{$(\pi,\pi')$ \text{\rm is special}}{$\mathit{flag} := 1$}
	\While{$\pi \neq \pi'$}{
		\If{\text{\rm there is no edge} $(\pi',\pi'')$ \text{\rm in} $\depg{\dep}$}{$\mathsf{reject}$}
		\textbf{guess} \text{an edge} $(\pi',\pi'')$ \text{in} $\depg{\dep}$\\
		\If{$(\pi',\pi'')$ \text{\rm is special}}{$\mathit{flag} := 1$}
		$\pi' := \pi''$
	}
	\If{$\mathit{flag} = 0$}{$\mathsf{reject}$}
	\textbf{guess} \text{a predicate} $R \in \sch{\dep}$ occurring in $D$\\
	\While{$R \neq P$}{
		\If{\text{\rm there is no edge} $(R,S)$ \text{\rm in} $\mathsf{pg}(\dep)$}{$\mathsf{reject}$}
		\textbf{guess} \text{an edge} $(R,S)$ \text{in} $\mathsf{pg}(\dep)$\\
		$R := S$
	}
	\Return{$\mathsf{Accept}$}
	\caption{$\mathsf{CheckWA}$}\label{alg:sl}
\end{algorithm}

\begin{lemma}\label{lem:nl-upper-bound}
	The problem of deciding whether $\dep$ is not $D$-weakly-acyclic is in \textsc{NL}.
\end{lemma}

\begin{proof}
	Checking whether $\dep$ is not $D$-weakly-acyclic can be done via the nondeterministic procedure $\mathsf{CheckWA}$, depicted in Algorithm~\ref{alg:sl}, which performs the following:
	%
	\begin{enumerate}
		\item a reachability check on the dependency graph $\depg{\dep}$ of $\dep$ in order to find a cycle $C$ with a special edge, and 
		\item a reachability check on the predicate graph of $\dep$, denoted $\mathsf{pg}(\dep)$, in order to check whether $C$ is $D$-supported -- recall that the nodes of $\mathsf{pg}(\dep)$ are the predicates of $\sch{\dep}$, and there is an edge $(R,P)$ iff there is a TGD $\sigma \in \dep$ such that $R$ occurs in $\body{\sigma}$ and $P$ occurs in $\head{\sigma}$. 
	\end{enumerate}
	It is easy to see that indeed $\mathsf{CheckWA}(D,\dep)$ accepts iff $\dep$ is not $D$-weakly-acyclic. 
	Furthermore, by performing the standard space complexity analysis for reachability-based procedures, we can show that each step of the computation of $\mathsf{CheckWA}(D,\dep)$ uses space
	\[
	O(\log |D| + \log ||\dep||),
	\]
	and the claim follows.
\end{proof}


Concerning the data complexity of $\mathsf{CT}(\class{SL})$, observe that even if we fix the set of TGDs, by the proof of Lemma~\ref{lem:nl-upper-bound}, $\mathsf{CheckWA}(D,\dep)$ uses space $O(\log |D|)$, and thus, it only shows that $\mathsf{CT}(\class{SL})$ is in \textsc{NL}. Therefore, to establish the desired \textsc{AC}$_0$ upper bound we need a finer approach.
We construct a UCQ $Q_\dep$ such that, for every database $D$, $\dep$ is not $D$-weakly-acyclic iff $D$ satisfies $Q_\dep$. Let $P_\dep$ be all the predicates $R$ of $\sch{\dep}$ such that there exists a position $(P,i)$ in $\depg{\dep}$ that belongs to a cycle with a special edge, and $R \reach{\dep} P$ (note that if $R \neq P$, $R \reach{\dep} P$ iff $P$ is reachable from $R$ in $\mathsf{pg}(\dep)$).
In simple words, $P_\dep$ collects all the potential reasons for which $\dep$ violates the condition of non-uniform weak-acyclicity. Thus, the UCQ $Q_\dep$ simply needs to check whether a potential reason (i.e., a predicate) occurs in the input database. Formally, $Q_\dep$ is defined as
\[
\bigvee_{R/n \in P_\dep} \exists x_{1}^{R} \cdots \exists x_{n}^{R} \, R(x_{1}^{R},\ldots,x_{n}^{R}),
\]
where $x_{1}^{R},\ldots,x_{n}^{R}$ are distinct variables.
It is clear that $Q_\dep$ is finite since $P_\dep$ is finite.
Consequently, for a fixed set $\dep \in \class{SL}$ of TGDs, we construct $Q_\dep$ in constant time, and then evaluate it over the input database $D$, which can be done in \textsc{AC}$_0$, and return $\mathsf{accept}$ if $D$ satisfies $Q_\dep$; otherwise, return $\mathsf{reject}$.


\section*{Proof of Proposition~\ref{pro:simplification}}

It is easy to see that $(2)$ implies $(1)$. The rest of the proof is devoted to showing statement $(2)$.
Let $D$ be a database, and $\dep \in \class{L}$. We first need a way to relate the terms in $\chase{D}{\dep}$ with the terms in $\chase{\simple{D}}{\simple{\dep}}$.
For brevity, given a linear TGD $\sigma = R(\bar x) \ra \exists \bar z\, \psi(\bar y,\bar z)$, where $\bar y \subseteq \bar x$, we write $\sigma_f$ for the simplification of $\sigma$ induced by the specialization $f$ of $\bar x$.

\begin{definition}\label{def:equiv-simpl}
	Consider the terms $t \in \adom{\chase{D}{\dep}}$ and $u \in \adom{\chase{\simple{D}}{\simple{\dep}}}$. We say that {\em $t$ is equivalent to $u$ up to simplification (w.r.t.~$D$ and $\dep$)}, denoted $t \equiv_\mathsf{sim} u$, if:
	\begin{enumerate}
		\item $t,u \in \ins{C}$ implies $t = u$.
		\item $t = \bot_{\sigma,h}^{z} \in \ins{N}$ and $u = \bot_{\sigma',h'}^{z'} \in \ins{N}$ implies
		\begin{enumerate}
			\item $\sigma' = \sigma_f$ and $z' = z$,
			\item $g = \{f(x) \mapsto h(x) \mid x \in \fr{\sigma}\}$ is a function of the form $\fr{\sigma'} \ra \adom{\chase{D}{\dep}}$, and
			\item for each $y \in \fr{\sigma'}$, $g(y) \equiv_{\mathsf{sim}} h'(y)$.
		\end{enumerate}
	\end{enumerate}
	We write $\mathsf{ES}(t)$ for the set of terms
	\begin{flalign*}
	&&\{u \in \adom{\chase{\simple{D}}{\simple{\dep}}} \mid t \equiv_\mathsf{sim} u\}.&&  \hfill\markfull
	\end{flalign*}
\end{definition}

We next establish a useful technical lemma concerning the notion of equivalence up to simplification.

\begin{lemma}\label{lem:simplification-aux-1}
	The following hold:
	\begin{enumerate}
		\item For each $t,t' \in \adom{\chase{D}{\dep}}$, $t \neq t'$ implies $\mathsf{ES}(t) \cap \mathsf{ES}(t') = \emptyset$.
		\item For each $t \in \adom{\chase{D}{\dep}}$ and $u \in \mathsf{ES}(t)$, $\depth{t} = \depth{u}$.
	\end{enumerate}
\end{lemma}

\begin{proof}
	\OMIT{
		(\textbf{Item 1}) We proceed by induction on the depth of $t$. 
		
		\medskip
		
		\noindent
		\textbf{Base Case:} If $\depth{t} = 0$, then $t$ is a constant and the only term equivalent to $t$ is $t$ itself, i.e., $\mathsf{ES}(t) = \{t\}$.
		
		\medskip
		
		\noindent
		\textbf{Induction Step:} Assume that $\depth{t} > 0$. Thus, $t$ is a null of the form $\skolem{\sigma}{h}{z}$, for some $\sigma \in \dep$, homomorphism $h$, and existentially quantified variable $z$ in $\sigma$. 
		By definition, for a null $u = \bot_{\sigma',h'}^{z}$ that occurs in $\chase{\simple{D}}{\simple{\dep}}$, $t \equiv_{\mathsf{sim}} u$ iff $\sigma' = \sigma_f$, $g = \{f(x) \mapsto h(x) \mid x \in \fr{\sigma}\}$ is a function, and for each $y \in \fr{\sigma'}$, $g(y) \equiv_{\mathsf{sim}} h'(y)$. 
		Note that there exist finitely many simplifications of $\sigma$ since $\arity{\dep}$ is finite.
		Moreover, for each $y \in \fr{\sigma'}$, by the definition of depth, $\depth{g(y)} < \depth{t}$. Hence, by induction hypothesis, $\mathsf{ES}(g(y))$ is finite. The latter, together with the fact that $\sigma$ has finitely many simplifications, imply that $\mathsf{ES}(t)$ is finite.
		
		\medskip
}
		(\textbf{Item 1}) By contradiction, assume that $t \neq t'$ and $\mathsf{ES}(t) \cap \mathsf{ES}(t') \neq \emptyset$. We show that $t = t'$, which contradicts our hypothesis.
		Let $u \in \mathsf{ES}(t) \cap \mathsf{ES}(t')$. We proceed by induction on the depth of $u$.
		
		\medskip
		
		\noindent
		\textbf{Base Case:} If $\depth{u} = 0$, then $u$ is a constant. Therefore, $t=u$ and $t'=u$, i.e., $t = t'$, and the claim follows.

		\medskip
		
		\noindent
		\textbf{Induction Step:} Assume that $\depth{u} > 0$. Thus, $u$ is a null of the form $\skolem{\sigma'}{h'}{z}$ for some $\sigma' \in \simple{\dep}$, homomorphism $h'$, and existentially quantified variable $z$ in $\sigma'$.
		Hence, $t,t'$ are nulls of the form $\skolem{\sigma}{h_1}{z}$ and $\skolem{\sigma}{h_2}{z}$, respectively, such that $\sigma' = \sigma_f$.
		Since $t \equiv_{\mathsf{sim}} u$ and $t' \equiv_{\mathsf{sim}} u$, by definition, $g_1 = \{f(x) \mapsto h_1(x) \mid x \in \fr{\sigma}\}$ and $g_2 = \{f(x) \mapsto h_2(x) \mid x \in \fr{\sigma}\}$ are functions, and for each $y \in \fr{\sigma'}$, $g_1(y) \equiv_{\mathsf{sim}} h'(y)$ and $g_2(y) \equiv_{\mathsf{sim}} h'(y)$.
		Thus, by induction hypothesis, $g_1(y) = g_2(y)$. Since $g_1$ and $g_2$ are functions, we conclude that $h_1 = h_2$, which in turn implies that $t = t'$.

		\medskip
		
		(\textbf{Item 2}) It easily follows by definition.
\end{proof}

We proceed to transfer the notion of equivalence up to simplification to tuples of terms and atoms.
Given two tuples $ \bar t = (t_1,\ldots,t_n)$ and $\bar u = (u_1,\ldots,u_m)$, for $n,m>0$, we write $\bar t \simeq \bar u$ if $n=m$, and $t_i = t_j$ iff $u_i = u_j$ for each $i,j \in [n]$.
Furthermore, assuming $n=m$, we write $\bar t \equiv_{\mathsf{sim}} \bar u$ for the fact that $t_i \equiv_{\mathsf{sim}} u_i$ for each $i \in [n]$.

\begin{definition}\label{def:atom-equivalence}
	Consider the atoms $\alpha = R(\bar t) \in \chase{D}{\dep}$ and $\beta = R_{(i_1,\ldots,i_{|\bar t|})}(\bar u) \in \chase{\simple{D}}{\simple{\dep}}$. We say that {\em $\alpha$ is equivalent to $\beta$ up to simplification (w.r.t.$D$ and $\dep$)}, denoted $\alpha \equiv_{\mathsf{sim}} \beta$, if $(i_1,\ldots,i_{|\bar t|}) = \id{}{\bar t}$, $\unique{\bar t} \simeq \bar u$, and $\unique{\bar t} \equiv_{\mathsf{sim}} \bar u$. 
	We write $\mathsf{ES}(\alpha)$ for $\{\beta \in \chase{\simple{D}}{\simple{\dep}} \mid \alpha \equiv_{\mathsf{sim}} \beta\}$. \hfill\markfull
\end{definition}

We are now ready to establish a crucial technical lemma, which is analogous to Lemma~\ref{lem:simplification-aux-1}, but for atoms instead of terms, that will easily lead to Proposition~\ref{pro:simplification}:

\begin{lemma}\label{lem:simplification-aux-2}
	The following hold:
	\begin{enumerate}
		\item $\{\mathsf{ES}(\alpha) \mid \alpha \in \chase{D}{\dep}\}$ forms a partition of $\chase{\simple{D}}{\simple{\dep}}$.
		\item For each $\alpha \in \chase{D}{\dep}$ and $\beta \in \mathsf{ES}(\alpha)$, $\depth{\alpha} = \depth{\beta}$.
	\end{enumerate}
\end{lemma}

\begin{proof}
	\OMIT{
	(\textbf{Item 1}) Consider an atom $R(\bar t) \in \chase{D}{\dep}$. By definition, each $\beta \in \mathsf{ES}(R(\bar t))$ is of the form $R_{\id{}{\bar t}}(\bar u)$ such that $\unique{\bar t} \equiv_{\mathsf{sim}} \bar u$. Since, by item (1) of Lemma~\ref{lem:simplification-aux-1}, for each term $t$ in $\unique{\bar{t}}$, $\mathsf{ES}(t)$ is finite, we can conclude that there are finitely many tuples $\bar u$ such that $\unique{\bar t} \equiv_{\mathsf{sim}} \bar u$. Therefore, $\mathsf{ES}(R(\bar t))$ is finite, as needed.

	\medskip
}
	(\textbf{Item 1}) To establish the claim we need to show:
	\begin{enumerate}
		\item for $\alpha,\alpha' \in\chase{D}{\dep}$, $\alpha \neq \alpha'$ implies $\mathsf{ES}(\alpha) \cap \mathsf{ES}(\alpha') = \emptyset$,
		\item for each $\alpha \in \chase{D}{\dep}$, $\mathsf{ES}(\alpha) \neq \emptyset$, and
		\item $\bigcup_{\alpha \in \chase{D}{\dep}} \mathsf{ES}(\alpha) = \chase{\simple{D}}{\simple{\dep}}$.
	\end{enumerate}
	We proceed to establish the above three statements:
	\begin{description}
		\item[Statement (a).] By contradiction, assume that $\alpha \neq \alpha'$ and $\mathsf{ES}(\alpha) \cap \mathsf{ES}(\alpha') \neq \emptyset$. We show that $\alpha = \alpha'$, which contradicts the fact that $\alpha \neq \alpha'$.
		Let $\beta \in \mathsf{ES}(\alpha) \cap \mathsf{ES}(\alpha')$, i.e., $\alpha \equiv_{\mathsf{sim}} \beta$ and $\alpha' \equiv_{\mathsf{sim}} \beta$. Assuming that $\beta$ is of the form $R_{(i_1,\ldots,i_n)}(\bar u)$, we get that $\alpha,\alpha'$ are of the form $R(\bar t)$ and $R(\bar t')$, respectively, and $\id{}{\bar t} = \id{}{\bar t'}$. Furthermore, $\unique{\bar t} \equiv_{\mathsf{sim}} \bar u$ and $\unique{\bar t'} \equiv_{\mathsf{sim}} \bar u$. By item (1) of Lemma~\ref{lem:simplification-aux-1}, $\unique{\bar t} = \unique{\bar t'}$. Since $\id{}{\bar t} = \id{}{\bar t'}$ and $\unique{\bar t} = \unique{\bar t'}$, we conclude that $\bar t = \bar t'$, which in turn implies that $\alpha = \alpha'$.
		
		\item[Statement (b).] Consider a (possibly infinite) valid chase derivation $I_0,I_1,\ldots$ of $D$ w.r.t.\ $\dep$, with $I_i \app{\sigma_i}{h_i} I_{i+1}$ for $i \ge 0$. We show that, for each $i \ge 0$, and $\alpha \in I_i$, $\mathsf{ES}(\alpha) \neq \emptyset$. We proceed by induction on the number of chase steps in the derivation. 
		
		\medskip
		
		\noindent
		\textbf{Base Case:}
		For each $\alpha \in I_0 = D$, it is clear that the atom $\simple{\alpha}$ belongs to $\mathsf{ES}(\alpha)$, since $\simple{D} \subseteq \chase{\simple{D}}{\simple{\dep}}$.
		
		\medskip
		
		\noindent
		\textbf{Induction Step:}
		Let $i  \ge 0$, and consider an atom $\alpha \in I_{i+1}$. 
		In case $\alpha \in I_{i}$, the claim follows by induction hypothesis.
		The interesting case is when $\alpha \in I_{i+1} \setminus I_{i}$. Therefore, $\alpha \in \result{\sigma_i}{h_i}$, i.e., $\alpha$ belongs to the result of the application of $(\sigma_i,h_i)$ to $I_i$. Let $\beta = R(\bar t) = h_i(\body{\sigma_i}) \in I_i$; recall that $\sigma_i$ is linear, and thus, $\body{\sigma_i}$ consists of a single atom.
		
		By induction hypothesis, there is an atom $\beta' \in \mathsf{ES}(\beta)$, which is of the form $R_{\id{}{\bar t}}(\bar u)$. Let $\body{\sigma_i} = R(\bar x)$, and let $\sigma_i'$ be the simplification of $\sigma_i$ induced by {\em the} specialization $f$ of $\bar x$ such that $h_i = \mu \circ f$ with $\mu$ being a bijection. The body of $\sigma_i'$ is of the form $R_{\id{}{\bar t}}(f(\bar x))$.
		%
		Since $\beta \equiv_{\mathsf{sim}} \beta'$, by definition, $\unique{\bar t} \simeq \bar u$. Hence, there exists a bijection $g : \unique{\bar t} \ra \bar u$ such that $g(\unique{\bar t}) = \bar u$.
		We can compose $\mu$ and $g$, and obtain a homomorphism $h_i'$ from $R_{\id{}{\bar t}}(f(\bar x))$ to $\beta'$, i.e., with $h_i' = g \circ \mu$, $R_{\id{}{\bar t}}(h_i'(f(\bar x))) = \beta'$.
		Therefore, there must be an atom $\alpha' \in \chase{\simple{D}}{\simple{\dep}}$ that belongs to $\result{\sigma_i'}{h_i'}$, which, by construction of $h_i'$ and $\sigma_i'$, is such that $\alpha \equiv_{\mathsf{sim}} \alpha'$. Hence, $\alpha' \in \mathsf{ES}(\alpha)$.
		
		\item[Statement (c).] We actually need to show that 
		\[
		\bigcup_{\alpha \in \chase{D}{\dep}} \mathsf{ES}(\alpha) \supseteq \chase{\simple{D}}{\simple{\dep}}. 
		\]
		This can be done via an inductive argument analogous to the one given above for statement (b), with the difference that the induction is on the number of chase steps of a valid chase derivation of $\simple{D}$ w.r.t.\ $\simple{\dep}$.
	\end{description}

	\medskip
	
	(\textbf{Item 2}) This follows from item (2) of Lemma~\ref{lem:simplification-aux-1}.
\end{proof}

It is not difficult to verify that statement $(2)$ of Proposition~\ref{pro:simplification} follows from Lemma~\ref{lem:simplification-aux-2}.



\section*{Proof of Theorem~\ref{the:lower-bound-linear}}

For each $\ell > 0$, we define the database
\[
D_\ell\ =\ \{R_0(c_1),\ldots,R_0(c_\ell)\}.
\]
We proceed to define $\dep_{n,m}$, for each $n,m>0$. 
For brevity, given a variable $x$, we simply write $x^k$ for $\underbrace{x,\ldots,x}_{k}$.
Let
\[
\dep_{n,m}\ =\ \dep_{\text{start}}\ \cup\ \bigcup_{i=1}^{n} \dep_{i}^{\forall}\ \cup\ \bigcup_{i=1}^{n-1} \dep_{i}^{\exists},
\]
where $\dep_{\text{start}}$ consists of the TGD
\[
R_0(x) \ra \exists y \exists z \, R_0(x),R_1(y^m,y,z,y), 
\]
$\dep_{i}^{\forall}$ contains, for each $j \in \{0,\ldots,m-1\}$, the TGD
\begin{multline*}
R_i(x_1,\ldots,x_{m-j-1},y,z^j,y,z,u)\ \ra\\ \exists v \exists w \, R_i(x_1,\ldots,x_{m-j-1},y,z^j,y,z,u),\\
R_i(x_1,\ldots,x_{m-j-1},z,y^j,y,z,v),\\
R_i(x_1,\ldots,x_{m-j-1},z,y^j,y,z,w)
\end{multline*}
and $\dep_{i}^{\exists}$ consists of the TGD
\[
R_i(x^m,y,x,z)\ \ra\
\exists v \exists w \, R_i(x^m,y,x,z),R_{i+1}(v^m,v,w,v).
\]
This completes the construction of $\dep_{n,m}$. It can be verified that, for each $\ell,n,m>0$, $\dep_{n,m} \in \class{SL} \cap \class{CT}_{D_{\ell}}$. It remains to show that $|\chase{D_{\ell}}{\dep_{n,m}}| \geq \ell \cdot 2^{n \cdot (2^m-1)}$. 
Roughly speaking, starting from an atom $R_0(c)$, for $c \in \ins{C}$, the TGD of $\dep_{\text{start}}$ generates an atom
\[
R_1(\underbrace{0,\ldots,0}_m,0,1,0).
\]
Note that the values $0$ and $1$ in the actual chase are represented as nulls, but for the sake of the discussion we see them as bits.
Now, starting from an atom of the form
\[
R_i(\underbrace{0,\ldots,0}_m,0,1,0),
\]
for $i \in [n]$, $\dep_{i}^{\forall}$ constructs a perfect binary tree of height $2^m - 1$ such that its $j$-th level consists of $2^j$ atoms of the form 
\[
R_i(b_1,\ldots,b_m,0,1,\bot),
\]
where $b_1,\ldots,b_m$ is the binary encoding of the number $j$, and $\bot$ is a null that occurs only in this atom. Therefore, we have $2^{(2^m-1)}$ leaves that are atoms of the form
\[
R_i(1,\ldots,1,0,1,\bot).
\]
Now, for each such $R_i$-leaf-atom, for $i \in [n-1]$, the TGD of $\dep_{i}^{\exists}$ generates an atom of the form 
\[
R_{i+1}(\underbrace{0,\ldots,0}_m,0,1,0), 
\]
which in turn will trigger $\dep_{i+1}^{\forall}$, and build another prefect binary tree with $R_{i+1}$-atoms as its nodes.

From the above discussion, it is not difficult to see that the following holds, which can be shown by induction on $i \in [n]$:

\begin{claim}\label{cla:lower-bound-linear}
	Let $D = \{R_0(c)\}$, where $c \in \ins{C}$. For each $i \in [n]$,
	\[
	|\{\bar t \mid R_i(\bar t) \in \chase{D}{\dep_{n,m}}\}|\ \geq\ 2^{i \cdot (2^m-1)}.
	\]
\end{claim}

By Claim~\ref{cla:lower-bound-linear}, it is straightforward to see that
\[
|\{\bar t \mid R_{n}(\bar t) \in \chase{D_\ell}{\dep_{n,m}}\}|\ \geq\ \ell \cdot 2^{n \cdot (2^m-1)},
\]
which in turn implies that
\[
|\chase{D_{\ell}}{\dep_{n,m}}|\ \geq\ \ell \cdot 2^{n \cdot (2^m-1)}.
\]


\section*{Proof of Theorem~\ref{the:complexity-linear}}

Consider a database $D$, and a set $\dep \in \class{L}$. By Theorem~\ref{thm:characterization-linear}, we know that $\dep \in \class{CT}_D$ iff $\simple{\dep}$ is $\simple{D}$-weakly-acyclic. 
Since co\textsc{PSpace} = \textsc{PSpace} and co\textsc{NL} = \textsc{NL}, for the \textsc{PSpace} and \textsc{NL }upper bounds claimed in Theorem~\ref{the:complexity-linear}, it suffices to show the following:

\begin{lemma}\label{lem:nl-upper-bound-linear}
	The problem of deciding whether $\simple{\dep}$ is not $\simple{D}$-weakly-acyclic is in \textsc{PSpace}, and in \textsc{NL} for schemas of bounded arity.
\end{lemma}

\begin{proof}[Proof (Sketch)]
	The problem of checking whether $\simple{\dep}$ is not $\simple{D}$-weakly-acyclic can be solved via an adapted version of the nondeterministic procedure $\mathsf{CheckWA}$, depicted in Algorithm~\ref{alg:sl}, which performs the following:
	%
	\begin{enumerate}
		\item a reachability check on the graph $\depg{\simple{\dep}}$ in order to find a cycle $C$ with a special edge, and 
		\item a reachability check on the predicate graph of $\simple{\dep}$ in order to check whether $C$ is $\simple{D}$-supported. 
	\end{enumerate}
	The difference is that the set $\simple{\dep}$ cannot be explicitly constructed since, in general, is of exponential size. Instead, we build it ``on the fly'' by simply guessing a TGD $\sigma$, an atom $R(\bar x) \in \body{\sigma}$, an atom $P(\bar y) \in \head{\sigma}$, and a specialization $f$ of $\bar x$. We can then construct the atoms $\simple{R(f(\bar x))}$ and  $\simple{P(f(\bar u))}$, which in turn will allow us to guess an edge in the dependency graph of $\simple{\dep}$.
	Concerning the check for $\simple{D}$-supportedness, we can simply iterate over the atoms of $D$ and compute, for each $R(\bar t) \in D$, the tuple $\id{}{\bar t}$, and then perform the second reachability check starting from the predicate $R_{\id{}{\bar t}}$.
	It is not difficult to show that each step of this adapted procedure uses space
	\[
	O(\log |D| + \arity{\dep} \cdot \log ||\dep||),
	\]
	and the claim follows.
\end{proof}

Concerning the data complexity of $\mathsf{CT}(\class{L})$, observe that even if we fix the set of TGDs, by the proof Lemma~\ref{lem:nl-upper-bound-linear}, we need space $O(\log |D|)$, and thus, it only shows that $\mathsf{CT}(\class{L})$ is in \textsc{NL}. Therefore, to establish the desired \textsc{AC}$_0$ upper bound in data complexity we need a finer approach.
We construct a UCQ $Q_\dep$ such that, for every database $D$, $\simple{\dep}$ is not $\simple{D}$-weakly-acyclic iff $D$ satisfies $Q_\dep$. 
The construction of $Q_\dep$ is similar in spirit to that of the UCQ in the case of simple linear TGDs.
Let $P_{\simple{\dep}}$ be the predicates $R$ of $\sch{\simple{\dep}}$ such that there is a position $(P,i)$ in $\depg{\simple{\dep}}$ that belongs to a cycle with a special edge, and $R \reach{\simple{\dep}} P$.
We then define $Q_\dep$ as the UCQ
\begin{multline*}
\bigvee_{\stackrel{R_{\bar \ell} \in P_\simple{\dep}}{\text{with } \bar \ell = (\ell_1,\ldots,\ell_n)}} \exists x_{1}^{R_{\bar \ell}} \cdots \exists x_{n}^{R_{\bar \ell}} \, \bigg( R(x_{1}^{R_{\bar \ell}},\ldots,x_{n}^{R_{\bar \ell}})\ \wedge\\
\bigwedge_{\{i,j \in [n] \mid \ell_i = \ell_j\}} \left(x_{i}^{R_\ell} = x_{j}^{R_\ell}\right)\bigg),
\end{multline*}
where $x_{1}^{R_{\bar \ell}},\ldots,x_{n}^{R_{\bar \ell}}$ are distinct variables.
It is clear that $Q_\dep$ is finite since $P_{\simple{\dep}}$ is finite.
It is also easy to verify that $Q_\dep$ is correct since, for a database $D$, a fact with predicate $R_{(\ell_1,\ldots,\ell_n)}$ belongs to $\simple{D}$ iff there exists an atom $R(t_1,\ldots,t_n) \in D$ such that, for each $i,j \in [n]$, $t_i = t_j$ iff $\ell_i = \ell_j$.
Consequently, for a fixed set $\dep \in \class{L}$ of TGDs, we construct $Q_\dep$ in constant time, and then evaluate it over the input database $D$, which can be done in \textsc{AC}$_0$, and return $\mathsf{accept}$ if $D$ satisfies $Q_\dep$; otherwise, return $\mathsf{reject}$.


\section*{Linearization}

We formally define the technique of linearization, which converts a set of guarded TGDs into a set of linear TGDs. In particular, we are going to explain how we convert a database $D$ and a set $\dep \in \class{G}$ into a database $\mathsf{lin}(D)$ and a set $\mathsf{lin}(\dep)$ of linear TGDs, respectively, such that Proposition~\ref{pro:linearization} holds, namely
\begin{enumerate}
	\item $\dep \in \class{CT}_{D}$ iff $\mathsf{lin}(\dep) \in \class{CT}_{\mathsf{lin}(D)}$.
	\item $\mathsf{maxdepth}(D,\dep) = \mathsf{maxdepth}(\mathsf{lin}(D),\mathsf{lin}(\dep))$.
\end{enumerate}
The proof of Proposition~\ref{pro:linearization} is given in the next section.
Before defining $\mathsf{lin}(D)$ and $\mathsf{lin}(\dep)$, we need some auxiliary notions.
In the rest, let $D$ be a database, and $\dep \in \class{G}$ a set of TGDs.

\medskip

\noindent
\textbf{Auxiliary Notions.}
For an atom $\alpha \in \chase{D}{\dep}$, the {\em type of $\alpha$ w.r.t.~$D$ and $\dep$}, denoted $\type{D,\dep}{\alpha}$, is the set of atoms
\[
\{\beta \in \chase{D}{\dep} \mid \adom{\beta} \subseteq \adom{\alpha}\}.
\]
The {\em completion of $D$ w.r.t.~$\dep$}, denoted $\completion{D}{\dep}$, is
\[
\{ \alpha \in \chase{D}{\dep} \mid \adom{\alpha} \subseteq \adom{D}\}.
\]
For a schema $\ins{S}$, and a set of terms $V$, we write $\base{\ins{S},V}$ for the set of atoms that can be formed using predicates from $\ins{S}$ and terms from $V$, that is, the set of atoms
\[
\left\{R(\bar v) \mid R \in \ins{S} \text{ and } \bar v \in V^{\arity{R}}\right\}.
\] 
A {\em $\dep$-type} $\tau$ is a pair $(\alpha,T)$, where $\alpha$ is an atom $R(\bar t) = R(t_1,\ldots,t_n)$, with $R \in \sch{\dep}$, $t_1 = 1$,
$
t_1 \leq t_i \leq \max_{j \in [i-1]}\{t_j\}+1
$
for $i \in \{2,\ldots,n\}$, and $T \subseteq \base{\sch{\dep},\adom{\alpha}} \setminus \{\alpha\}$.
We write $\guard{\tau}$ for the atom $\alpha$, $\atoms{\tau}$ for the set $T \cup \{\alpha\}$, and $\arity{\tau}$ for the integer $\arity{\guard{\tau}}$, i.e., the {\em arity of $\tau$}. Intuitively, $\tau$ encodes the shape of a guard and its type.
%
%
Given a tuple $\bar u$ such that $\bar u \simeq \bar t$, the {\em instantiation of $\tau$ with $\bar u$}, denoted $\tau(\bar u)$, is the set of atoms obtained from $\atoms{\tau}$ after replacing $t_i$ with $u_i$.
Finally, for $\alpha = R(\bar t)$, we write $\types{\dep}{\alpha}$ for the set of all $\dep$-types $\tau$ such that $\guard{\tau} = R(\bar u)$, and
$\bar t \simeq \bar u$.

\medskip

\noindent
\textbf{Database Linearization.} We define the linearization of $D$ as
\[
\begin{array}{rcl}
\ling{D} &=& \left\{[\tau](\bar t)\ ~\left|~
\begin{array}{l}
R(\bar t) \in D,\\
\tau \in \types{\dep}{R(\bar t)},\\
\tau(\bar t) = \type{D,\dep}{R(\bar t)}
\end{array} \right.\right\},
\end{array}
\]
where $[\tau]$ is a new predicate not occurring in $\sch{\dep}$.
Here is a simple example that illustrates the above definition

\begin{example}\label{ex:db-lin}
	Consider the database
	\[
	D\ =\ \{R(a,a,b,c)\}
	\]
	and the set $\dep$ consisting of the guarded TGDs
\[
	\sigma\ =\ P(x,y,x,u,w),S(x,u)\ \ra\
	\exists z_1 \exists z_2 \, R(u,y,x,z_1),T(z_1,z_2,x)
\]
	and
	\[
	\sigma'\ =\ R(x,x,y,z) \ra Q(x,z).
	\]
	It is clear that the $\dep$-type
	\[
	\tau\ =\ (R(1,1,2,3), \{Q(1,3)\})
	\]
	belongs to $\types{\dep}{R(a,a,b,c)}$ since it holds that $(a,a,b,c) \simeq (1,1,2,3)$.
	Moreover, it is clear that
\[
		\tau(a,a,b,c)\ =\ \{R(a,a,b,c),Q(a,c)\}\ =\ \type{D,\dep}{R(a,a,b,c)}.
\]
	We conclude that $\ling{D} = \{\predt{\tau}(a,b,c)\}$. \hfill\markfull
\end{example}


\noindent
\textbf{TGD Linearization.} 
Consider a  TGD $\sigma \in \dep$ of the form
\[
\varphi(\bar x,\bar y)\ \ra\ \exists z_1 \cdots \exists z_n \, R_1(\bar u_1),\ldots,R_m(\bar u_m)
\]
with $\guard{\sigma} = R(\bar u)$, and $(\bar x \cup \{z_i\}_{i \in [n]})$ be the variables occurring in $\head{\sigma}$.
For a $\dep$-type $\tau$ such that there exists a homomorphism $h$ from $\varphi(\bar x,\bar y)$ to $\atoms{\tau}$ and $h(R(\bar u)) = \guard{\tau}$, the {\em linearization of $\sigma$ induced by $\tau$ and $h$} is
\[
[\tau](\bar u)\ \ra\ \exists z_1 \cdots \exists z_n \, [\tau_1](\bar u_1),\ldots,[\tau_m](\bar u_m),
\]
where, for each $i \in [m]$, $[\tau_i]$ is defined as follows.
Let $f$ be the function from the variables in $\head{\sigma}$ to $\mathbb{N}$  defined as
\begin{eqnarray*}
	f(t)\
	=\ \left\{
	\begin{array}{ll}
		h(t) & \quad \text{if } t \in \bar x\\
		&\\
		\arity{\dep}+i & \quad \text{if } t = z_i.
	\end{array}\right.
\end{eqnarray*}
%
With $\alpha_i = R_i(f(\bar u_i))$, $\widehat{\tau_i} = (\alpha_i,T_i)$, where $T_i$ is
\[
\{\beta \in \completion{I}{\dep} \mid \adom{\beta} \subseteq \adom{\alpha_i}\} \setminus \{\alpha_i\}
\]
with
\[
I\ =\ \{\alpha_1,\ldots,\alpha_m\}\ \cup\ \atoms{\tau}.
\]
Note that $\widehat{\tau_i}$ is not a proper $\dep$-type since the integers in $\alpha_i$ do not appear in the right order.
Let $\rho$ be the renaming function that renames the integers in $\alpha_i$ in order to appear in increasing order starting from $1$ (e.g., $\rho(R(2,2,4,1)) = R(1,1,2,3)$; the formal definition is omitted since it is clear what the function $\rho$ does). We then define $\tau_i$ as the pair $(\rho(\alpha_i),\rho(T_i))$.
We write $\ling{\sigma}$ for the set of all linearizations of $\sigma$, induced by some $\dep$-type $\tau$, and homomorphism $h$ from $\varphi(\bar x,\bar y)$ to $\atoms{\tau}$ such that $h(R(\bar u)) = \guard{\tau}$.
Finally, the \emph{linearization of $\dep$} is defined as the (finite) set of TGDs
\[
\ling{\dep}\ =\ \bigcup\limits_{\sigma \in \dep} \ling{\sigma}.
\]
An example that illustrates the above definition follows.

\begin{example}
	Let $\dep$ be the set of TGDs given in Example~\ref{ex:db-lin}. Consider the $\dep$-type
	\[
	\tau\ = (P(1,2,1,2,3), \{S(1,2),S(1,1)\}).
	\]
	It is easy to verify that 
	\[
	h\ =\ \{ x \mapsto 1, y \mapsto 2, u \mapsto 2, w \mapsto 3\}
	\]
	is a homomorphism from $\body{\sigma}$ to
	\[
	\atoms{\tau}\ =\ \{P(1,2,1,2,3),S(1,2),S(1,1)\}, 
	\]
	and $h(P(x,y,x,u,w)) = P(1,2,1,2,3)$.
	The linearization of $\sigma$ induced by $\tau$ and $h$ is the linear TGD
	\[
	\predt{\tau}(x,y,x,u,w) \ra \exists z_1, \exists z_2\, \predt{\tau_1}(u,y,x,z_1), \predt{\tau_2}(z_1,z_2,x),
	\]
	where 
	\begin{eqnarray*}
		\tau_1 &=& (R(1,1,2,3),\{S(2,1),S(2,2),Q(1,3)\}),\\
		\tau_2 &=& (T(1,2,3), \emptyset). \hspace{54mm} \hfill\markfull
	\end{eqnarray*}
\end{example}


\section*{Proof of Proposition~\ref{pro:linearization}}

It is easy to see that $(2)$ implies $(1)$. The rest of the proof is devoted to showing statement $(2)$.
Let $D$ be a database, and $\dep \in \class{G}$ a set of TGDs. We first need a way, analogous to that of equivalence up to simplification (given in Definition~\ref{def:equiv-simpl}), to relate the terms in $\chase{D}{\dep}$ with the terms in $\chase{\ling{D}}{\ling{\dep}}$.
For a guarded TGD $\sigma$, we write $\sigma_{\tau,g}$, where $\tau$ is a $\dep$-type, and $g$ a homomorphism from $\body{\sigma}$ to $\atoms{\tau}$ with $g(\guard{\sigma}) = \guard{\tau}$, for the linearization of $\sigma$ induced by $\tau$ and $g$.

\begin{definition}\label{def:equiv-linearization}
	Consider the terms $t \in \adom{\chase{D}{\dep}}$ and $u \in \adom{\chase{\ling{D}}{\ling{\dep}}}$. We say that {\em $t$ is equivalent to $u$ up to linearization (w.r.t.~$D$ and $\dep$)}, denoted $t \equiv_\mathsf{lin} u$, if:
	\begin{enumerate}
		\item $t,u \in \ins{C}$ implies $t = u$.
		\item $t = \bot_{\sigma,h}^{z} \in \ins{N}$ and $u = \bot_{\sigma',h'}^{z'} \in \ins{N}$ implies
		\begin{enumerate}
			\item $\sigma' = \sigma_{\tau,g}$ and $z' = z$, and
			\item for each $y \in \fr{\sigma'}$, $h(y) \equiv_{\mathsf{lin}} h'(y)$.
		\end{enumerate}
	\end{enumerate}
	We write $\mathsf{EL}(t)$ for the set of terms
	\begin{flalign*}
		&& \{u \in \adom{\chase{\ling{D}}{\ling{\dep}}} \mid t \equiv_\mathsf{lin} u\}. && \markfull
	\end{flalign*}
\end{definition}

We next establish a useful lemma, analogous to Lemma~\ref{lem:simplification-aux-1}, concerning the notion of equivalence up to linearization.

\begin{lemma}\label{lem:linearization-aux-1}
	The following hold:
	\begin{enumerate}
		\item For each $t,t' \in \adom{\chase{D}{\dep}}$, $t \neq t'$ implies $\mathsf{EL}(t) \cap \mathsf{EL}(t') = \emptyset$.
		\item For each $t \in \adom{\chase{D}{\dep}}$ and $u \in \mathsf{EL}(t)$, $\depth{t} = \depth{u}$.
	\end{enumerate}
\end{lemma}

\begin{proof}
	\OMIT{
	(\textbf{Item 1}) We proceed by induction on the depth of $t$. 
	
	\medskip
	
	\noindent
	\textbf{Base Case:} If $\depth{t} = 0$, then $t$ is a constant and the only term equivalent to $t$ is $t$ itself, i.e., $\mathsf{EL}(t) = \{t\}$.
	
	\medskip
	
	\noindent
	\textbf{Induction Step:} Assume that $\depth{t} > 0$. Thus, $t$ is a null of the form $\skolem{\sigma}{h}{z}$, for some $\sigma \in \dep$, homomorphism $h$, and existentially quantified variable $z$ in $\sigma$. 
	By definition, for a null $u = \bot_{\sigma',h'}^{z}$ that occurs in $\chase{\ling{D}}{\ling{\dep}}$, $t \equiv_{\mathsf{lin}} u$ iff $\sigma' = \sigma_{\tau,g}$, and for each $y \in \fr{\sigma'}$, $h(y) \equiv_{\mathsf{lin}} h'(y)$. 
	Note that there exist finitely many linearizations of $\sigma$ since $\arity{\dep}$ and $|\sch{\dep}|$ are finite. 
	Moreover, for each $y \in \fr{\sigma'}$, by definition, $\depth{h(y)} < \depth{t}$. Hence, by induction hypothesis, $\mathsf{EL}(h(y))$ is finite. This, together with the fact that $\sigma$ has finitely many linearizations, implies $\mathsf{EL}(t)$ is finite.

	\medskip
	
}
	(\textbf{Item 1}) By contradiction, assume that $t \neq t'$ and $\mathsf{EL}(t) \cap \mathsf{EL}(t') \neq \emptyset$. We show that $t = t'$, which contradicts our hypothesis.
	Let $u \in \mathsf{EL}(t) \cap \mathsf{EL}(t')$. We proceed by induction on the depth of $u$.
	
	\medskip
	
	\noindent
	\textbf{Base Case:} If $\depth{u} = 0$, then $u$ is a constant. Therefore, $t=u$ and $t'=u$, i.e., $t = t'$, and the claim follows.

	\medskip
	
	\noindent
	\textbf{Induction Step:} Assume that $\depth{u} > 0$. Thus, $u$ is a null of the form $\skolem{\sigma'}{h'}{z}$ for some $\sigma' \in \ling{\dep}$, homomorphism $h'$, and existentially quantified variable $z$ in $\sigma'$.
	Hence, $t,t'$ are nulls of the form $\skolem{\sigma}{h_1}{z}$ and $\skolem{\sigma}{h_2}{z}$, respectively, such that $\sigma' = \sigma_{\tau,g}$.
	Since $t \equiv_{\mathsf{lin}} u$ and $t' \equiv_{\mathsf{lin}} u$, by definition, $h_1(y) \equiv_{\mathsf{lin}} h'(y)$ and $h_2(y) \equiv_{\mathsf{lin}} h'(y)$, for each $y \in \fr{\sigma'}$. Thus, by induction hypothesis, $h_1(y) = h_2(y)$. Since $h_1$ and $h_2$ are substitutions from $\fr{\sigma'}$, we conclude that $h_1 = h_2$, which in turn implies that $t = t'$.
	
	\medskip
	
	(\textbf{Item 2}) It easily follows by definition.
\end{proof}

We now transfer the notion of equivalence up to linearization to tuples of terms and atoms.
Recall that, for two tuples $ \bar t = (t_1,\ldots,t_n)$ and $\bar u = (u_1,\ldots,u_m)$, for $n,m>0$, we write $\bar t \simeq \bar u$ if $n=m$, and $t_i = t_j$ iff $u_i = u_j$ for each $i,j \in [n]$.
Furthermore, assuming that $n=m$, we write $\bar t \equiv_{\mathsf{lin}} \bar u$ for the fact that $t_i \equiv_{\mathsf{lin}} u_i$ for each $i \in [n]$.
%

\begin{definition}\label{def:atom-equivalence}
	Consider the atoms $\alpha = R(\bar t) \in \chase{D}{\dep}$ and $\beta = [\tau](\bar u) \in \chase{\ling{D}}{\ling{\dep}}$. We say that {\em $\alpha$ is equivalent to $\beta$ up to linearization (w.r.t.~$D$ and $\dep$)}, denoted $\alpha \equiv_{\mathsf{lin}} \beta$, if there is a bijection $\mu$ from $\adom{\type{D,\dep}{\alpha}}$ to $\adom{\atoms{\tau}}$ such that $\mu(\type{D,\dep}{\alpha}) = \atoms{\tau}$ and $\mu(\alpha) = \guard{\tau}$, $\bar t \simeq \bar u$, and $\bar t \equiv_{\mathsf{lin}} \bar u$. 
	We let $\mathsf{EL}(\alpha) = \{\beta \in \chase{\ling{D}}{\ling{\dep}} \mid \alpha \equiv_{\mathsf{sim}} \beta\}$. \hfill\markfull
\end{definition}

We are now ready to establish a crucial technical lemma, which is analogous to Lemma~\ref{lem:linearization-aux-1}, but for atoms instead of terms, that will easily lead to Proposition~\ref{pro:linearization}:

\begin{lemma}\label{lem:linearization-aux-2}
	The following hold:
	\begin{enumerate}
		\item $\{\mathsf{EL}(\alpha) \mid \alpha \in \chase{D}{\dep}\}$ forms a partition of $\chase{\ling{D}}{\ling{\dep}}$.
		\item For each $\alpha \in \chase{D}{\dep}$ and $\beta \in \mathsf{EL}(\alpha)$, $\depth{\alpha} = \depth{\beta}$.
	\end{enumerate}
\end{lemma}

\begin{proof}
	\OMIT{
	(\textbf{Item 1}) Consider an atom $R(\bar t) \in \chase{D}{\dep}$. By definition, each $\beta \in \mathsf{EL}(R(\bar t))$ is of the form $[\tau](\bar u)$ such that $\bar t \equiv_{\mathsf{lin}} \bar u$. Since, by item (1) of Lemma~\ref{lem:linearization-aux-1}, for each term $t$ in $\bar t$, $\mathsf{EL}(t)$ is finite, we get that there are finitely many tuples $\bar u$ such that $\bar t \equiv_{\mathsf{lin}} \bar u$. Therefore, $\mathsf{EL}(R(\bar t))$ is finite, as needed.

\medskip
}
(\textbf{Item 1}) To establish the claim we need to show:
\begin{enumerate}
	\item for $\alpha,\alpha' \in\chase{D}{\dep}$, $\alpha \neq \alpha'$ implies $\mathsf{EL}(\alpha) \cap \mathsf{EL}(\alpha') = \emptyset$,
	\item for each $\alpha \in \chase{D}{\dep}$, $\mathsf{EL}(\alpha) \neq \emptyset$, and
	\item $\bigcup_{\alpha \in \chase{D}{\dep}} \mathsf{EL}(\alpha) = \chase{\ling{D}}{\ling{\dep}}$.
\end{enumerate}
We proceed to establish the above three statements:
\begin{description}
	\item[Statement (a).] By contradiction, assume that $\alpha \neq \alpha'$ and $\mathsf{EL}(\alpha) \cap \mathsf{EL}(\alpha') \neq \emptyset$. We show that $\alpha = \alpha'$, which contradicts the fact that $\alpha \neq \alpha'$.
	Let $\beta \in \mathsf{EL}(\alpha) \cap \mathsf{EL}(\alpha')$, i.e., $\alpha \equiv_{\mathsf{lin}} \beta$ and $\alpha' \equiv_{\mathsf{lin}} \beta$.
	Assuming that $\beta$ is of the form $\predt{\tau}(\bar u)$, with $\guard{\tau} = R(i_1,\ldots,i_k)$, we get that $\alpha,\alpha'$ are of the form $R(\bar t)$ and $R(\bar t')$, respectively. Furthermore, $\bar t \equiv_{\mathsf{lin}} \bar u$ and $\bar t' \equiv_{\mathsf{lin}} \bar u$. By item (1) of Lemma~\ref{lem:linearization-aux-1}, $\bar t = \bar t'$, and thus, $\alpha = \alpha'$.
	
	\item[Statement (b).] Consider a (possibly infinite) valid chase derivation $I_0,I_1,\ldots$ of $D$ w.r.t.\ $\dep$, with $I_i \app{\sigma_i}{h_i} I_{i+1}$ for $i \ge 0$. We show that, for each $i \ge 0$, and $\alpha \in I_i$, $\mathsf{EL}(\alpha) \neq \emptyset$. We proceed by induction on the number of chase steps in the derivation.
	
	\medskip
	
	\noindent
	\textbf{Base Case:}
	For each $\alpha \in I_0 = D$, $\mathsf{EL}(\alpha) \neq \emptyset$ follows from the fact that $\ling{D} \subseteq \chase{\ling{D}}{\ling{\dep}}$.
	
	\medskip
	
	\noindent
	\textbf{Induction Step:}
	Let $i  \ge 0$, and consider an atom $\alpha \in I_{i+1}$. 
	In case $\alpha \in I_{i}$, the claim follows by induction hypothesis.
	The interesting case is when $\alpha \in I_{i+1} \setminus I_{i}$. Therefore, $\alpha \in \result{\sigma_i}{h_i}$, i.e., $\alpha$ belongs to the result of the application of $(\sigma_i,h_i)$ to $I_i$. Let $\beta$ be the atom $R(\bar t) = h_i(\guard{\sigma_i}) \in I_i$.
	
	By induction hypothesis, there is an atom $\beta' = \predt{\tau}(\bar u) \in \mathsf{EL}(\beta)$, which implies that there exists a bijection $\mu$ from $\adom{\type{D,\dep}{\beta}}$ to $\adom{\atoms{\tau}}$ such that $\mu(\type{D,\dep}{\beta}) = \atoms{\tau}$ and $\mu(\beta) = \guard{\tau}$. Hence, the substitution $f = \mu \circ h_i$ is a homomorphism from $\body{\sigma_i}$ to $\atoms{\tau}$, with $f(\guard{\sigma_i}) = \guard{\tau}$. Let $\guard{\sigma_i} = R(\bar x)$, and let $\sigma_i'$ be the linearization of $\sigma_i$ induced by $\tau$ and $f$.
	The body of $\sigma_i'$ is of the form $\predt{\tau}(\bar x)$.
	Since $\beta \equiv_{\mathsf{lin}} \beta'$, by definition, $\bar t \simeq \bar u$. Hence, there exists a bijection $g : \bar t \ra \bar u$ such that $g(\bar t) = \bar u$.
	We can compose $h_i$ and $g$, and obtain an homomorphism $h_i'$ from $\predt{\tau}(\bar x)$ to $\beta'$, i.e., with $h_i' = g \circ h_i$,
	$\predt{\tau}(h_i'(\bar x)) = \beta'$. Therefore, there must be an atom $\alpha' \in \chase{\ling{D}}{\ling{\dep}}$ that belongs to $\result{\sigma_i'}{h_i'}$. We show that such $\alpha'$ belongs to $\mathsf{EL}(\alpha)$.

	Assume $\sigma_i$ and $\sigma_i'$ are of the form
	\[
	\varphi(\bar x) \ra \exists z_1 \cdots \exists z_n \, R_1(\bar u_1),\ldots,R_m(\bar u_m),
	\]
	and
	\[
	\predt{\tau}(\bar x) \ra \exists z_1 \cdots \exists z_n \, \predt{\tau_1}(\bar u_1),\ldots,\predt{\tau_m}(\bar u_m),
	\]
	respectively, 
	and assume $\alpha$ is the atom $\mu_i(R_\ell(\bar u_\ell)) = R_\ell(\bar k) \in \result{\sigma_i}{h_i}$, for some $\ell \in [m]$, where $\mu_i$ is the function such that $\mu_i(\head{\sigma_i}) = \result{\sigma_i}{h_i}$.
	Consider the atom $\alpha' = \mu_i'(\predt{\tau_\ell}(\bar u_\ell)) = \predt{\tau_\ell}(\bar v) \in \result{\sigma_i'}{h_i'}$, with $\mu_i'(\head{\sigma_i'}) = \result{\sigma_i'}{h_i'}$.
	Note that $\bar k \simeq \bar v$ and $\bar k \equiv_{\mathsf{lin}} \bar v$, since $\bar t \simeq \bar u$ and $\bar t \equiv_{\mathsf{lin}} \bar u$, respectively.
	It remains to prove that there is a bijection $\mu$ from $\adom{\type{D,\dep}{\alpha}}$ to $\adom{\atoms{\tau_\ell}}$ such that $\mu(\type{D,\dep}{\alpha}) = \atoms{\tau_\ell}$ and $\mu(\alpha) = \guard{\tau_\ell}$.
	
	Recalling that $f$ is the homomorphism from $\body{\sigma_i}$ to $\atoms{\tau}$, with $f(\guard{\sigma_i}) = \guard{\tau}$, let $f'$ be the function from the variables of $\head{\sigma_i}$ to $\mathbb{N}$, defined as $f'(t) = f(t)$, if $t \in \bar x$, and $f'(t) = \arity{\dep} + i$, if $t = z_i$.
	By definition of $\sigma_i'$, with $\alpha_j = f'(R_j(\bar u_j))$, for each $j \in [m]$, $\tau_\ell$ is the $\dep$-type obtained by renaming the integers of $\widehat{\tau_\ell} = (\alpha_\ell,T_\ell)$, where $\{\alpha_\ell\} \cup T_\ell$ is the set
	\[
	\{\gamma \in \completion{I}{\dep} \mid \adom{\gamma} \subseteq \adom{\alpha_\ell}\}
	\]
	with
	\[I = \{\alpha_1,\ldots,\alpha_m\}\ \cup\ \atoms{\tau}.
	\]
	Moreover, by Lemma 6 of~\cite{GoMP20}, we know that $\type{D,\dep}{\alpha}$ can be equivalently defined as
	\[
	\{ \gamma \in \completion{J}{\dep} \mid \adom{\gamma} \subseteq \adom{\alpha}\},
	\]
	with
	\[J = \result{\sigma_i}{h_i} \cup\type{D,\dep}{\beta}.
	\]
	By induction hypothesis, there is a bijection $\mu$ from $\adom{\type{D,\dep}{\beta}}$ to $\adom{\atoms{\tau}}$ such that 
	\[
	\mu(\type{D,\dep}{\beta})\ =\ \atoms{\tau}
	\] 
	and 
	\[\mu(\beta)\ =\ \guard{\tau}.
	\]
	By definition of $\result{\sigma_i}{h_i}$, and by construction of the atoms $\alpha_1,\ldots,\alpha_m$, $\mu$ can be extended to a bijection $\mu'$ from $\adom{J}$ to $\adom{I}$ such that $\mu'(J) = I$ and $\mu'(\alpha) = \alpha_\ell$.
	Furthermore, by definition of completion, we conclude that $\mu'$ is also a bijection from $\adom{\type{D,\dep}{\alpha}}$ to $\adom{\{\alpha_\ell\} \cup T_\ell}$ such that $\mu'(\type{D,\dep}{\alpha}) = \{\alpha_\ell\} \cup T_\ell$ and $\mu'(\alpha) = \alpha_\ell$.
	Since $\tau_\ell$ is obtained by applying a renaming function $\rho$ of the integers of $\alpha_\ell$ to $\widehat{\tau_\ell}$, then $\mu '' = \rho \circ \mu'$ is a bijection from $\adom{\type{D,\dep}{\alpha}}$ to $\adom{\atoms{\tau_\ell}}$ such that $\mu''(\type{D,\dep}{\alpha}) = \atoms{\tau_\ell}$ and $\mu''(\alpha) = \guard{\tau_\ell}$.
	
	\item[Statement (c).] We actually need to show that 
	\[
	\bigcup_{\alpha \in \chase{D}{\dep}} \mathsf{EL}(\alpha)\ \supseteq\ \chase{\ling{D}}{\ling{\dep}}. 
	\]
	This can be done via an inductive argument analogous to the one given above for statement (b), with the difference that the induction is on the number of chase steps of a valid chase derivation of $\ling{D}$ w.r.t.~$\ling{\dep}$.
\end{description}

\medskip

(\textbf{Item 2}) This follows from item (2) of Lemma~\ref{lem:linearization-aux-1}.
\end{proof}

It is not difficult to verify that statement $(2)$ of Proposition~\ref{pro:linearization} follows from Lemma~\ref{lem:linearization-aux-2}. 



\section*{Proof of Theorem~\ref{the:lower-bound-guarded}}

For each $\ell > 0$, we define the database
\[
D_\ell\ =\ \{\text{\rm Node}(c_1,c_1,0,1),\ldots,\text{\rm Node}(c_\ell,c_\ell,0,1)\}.
\]
An atom of the form $\text{\rm Node}(u,v,z,o)$ simply states that $v$ is a node ($v$ is the encoded node) with $u$ being its parent.
The repeated constant $c_i$ indicates that $c_i$ is a root note of a tree of the $0$-th stratum.
As we shall see, the constants $0$ and $1$ in the last two positions of the database atoms will give us access to $0$ and $1$, respectively, without having to explicitly mention them in the TGDs.

We proceed to define $\dep_{n,m}$, for each $n,m>0$. Intuitively speaking, the goal of $\dep_{n,m}$ is to construct, for each $i \in [\ell]$, $2^n$ strata such that the $0$-th stratum consists of a full binary tree of depth $2^{(2^m)}-1$ rooted at $c_i$, while the $j$-th stratum, for $j \in \{1,\ldots,2^n-1\}$, consists of $2^{\left(j \cdot \left(2^{(2^m)}-1\right)\right)}$ full binary trees of depth $2^{(2^m)}-1$ rooted at the leaf nodes of the trees of stratum $j-1$.
The crucial components of $\dep_{n,m}$ are the counting mechanisms for the strata (this is an exponential counter w.r.t.~$n$), and the depth of the trees (this is a double-exponential counter w.r.t.~$m$). To this end, we need a convenient representation of strata ids and depth ids.
Strata ids are represented as $n$-bit binary numbers. More precisely, we use $n$ binary predicates $S_1,\ldots,S_n$ with the following meaning: for $i \in [n]$, $S_i(x,b)$ encodes the fact that $x$ is a node of a binary tree of a stratum $b_1,\ldots,b_n$ with $b_i = b$. In other words, the set of atoms $S_1(x,b_1),\ldots,S_n(x,b_n)$ encodes the fact that $x$ is a node of a tree of the stratum $b_1,\ldots,b_n$.
Depth ids are represented as $2^m$-bit binary numbers. More precisely, we use an $(m+2)$-ary predicate $\text{\rm Depth}$ with the following meaning: $\text{\rm Depth}(x,b_1,\ldots,b_m,b)$ encodes the fact that $x$ is a node of depth $b'_1,\ldots,b'_{2^m}$ with $b'_i = b$ assuming that $b_1,\ldots,b_m$ is the binary encoding of the digit-id $i$.
The formal construction of $\dep_{n,m}$ follows. For brevity, given a variable $x$, we write $x^k$ for $\underbrace{x,\ldots,x}_{k}$.

The following TGD creates the root of the $0$-th stratum
\[
\text{\rm Node}(x,x,z,o)\ \ra\ \text{\rm Root}(x), S_1(x,z),\ldots,S_n(x,z).
\]
%
For implementing the double-exponential depth counter, we first associate digit-ids to tree nodes via the following TGDs
\[
\text{\rm Node}(x,y,z,o)\ \ra\ \text{\rm Did}(x,y,z,o,z^m)
\]
that creates the digit-id zero, and, for each $i \in [m]$,
\begin{align*}
\text{\rm Did}(x,y,z,o,w_1,\ldots,w_{i-1},z,w_{i+1},\ldots,w_m)\ \ra\\ \text{\rm Did}(x,y,z,o,w_1,\ldots,w_{i-1},o,w_{i+1},\ldots,w_m)
\end{align*}
that create all the other digit-ids starting from the digit-id zero. In simple words, for each tree node, the above rules generate all the $2^m$ digit-ids by starting from $0^m$, and successively switching a zero ($z$) to an one ($o$) in all possible ways.

Having the $2^m$ digit-ids for each node in place, we can then set the depth counter of a root node to zero by simply setting each bit of the $2^m$-bit depth counter to zero. This is easily done via the TGD
\[
\text{\rm Did}(x,y,z,o,w_1,\ldots,w_m), \text{\rm Root}(y)\ \ra\ \text{\rm Depth}(y,w_1,\ldots,w_m,z).
\]

As we shall see, for increasing the depth counter it is crucial to be able to say which digit-id comes after a given digit-id. In other words, we need a successor relation over digit-ids. This can be easily computed via the following set of TGDs: for each $i \in \{1,\ldots,m\}$
\begin{align*}
\text{\rm Did}(x,y,z,o,w_1,\ldots,w_{i-1},z,o^{m-i})\ \ra\\ \text{\rm Succ}(x,y,z,o,w_1,\ldots,w_{i-1},z,o^{m-i},w_1,\ldots,w_{i-1},o,z^{m-i}).
\end{align*}

We further need a way to specify whether a node belongs to a tree of the last stratum, and whether it has reached the maximal depth of its tree. We can actually easily specify the complement via the following TGDs: for each $i \in [n]$
\[
\text{\rm Node}(x,y,z,o),S_i(y,z)\ \ra\ \text{\rm NonMaxStratum}(y)
\]
and
\[
\text{\rm Depth}(x,w_1,\ldots,w_m,z)\ \ra \text{\rm NonMaxDepth}(x).
\]

We can now generate non-root nodes, i.e., nodes that belong to a tree of a certain stratum. As long as a node has not reached the maximal depth, it has two children generated via the TGD
\begin{align*}
\text{\rm Node}(x,y,z,o),\text{\rm NonMaxDepth}(y)\ \ra\\
\exists w \exists w' \, \text{\rm Node}(y,w,z,o),\text{\rm NonRoot}(w), \text{\rm Node}(y,w',z,o),\text{\rm NonRoot}(w')
\end{align*}
and we further specify that the newly generated nodes belong to the same stratum as their parent via the TGDs
\[
\text{\rm Node}(x,y,z,o),\text{\rm NonRoot}(y),S_i(x,z)\ \ra\ S_i(y,z)
\]
and
\[
\text{\rm Node}(x,y,z,o),\text{\rm NonRoot}(y),S_i(x,o)\ \ra\ S_i(y,o).
\]

We now need to specify that the depth of the newly generated nodes is the depth of their parent plus one, i.e., we need a mechanism for increasing the depth counter of the parent by one. This can be done in a standard way: convert the rightmost digit that is zero (which we call pivot) to one, change all the digits right to the pivot digit (which are by definition one) to zero, and keep all the digits left to the pivot digit unchanged. To this end, we need a way to classify digits as pivot, change and copy. This is done via the following TGDs; $\bar w$ and $\bar w'$ are $m$-ary tuples of distinct variables:
\begin{align*}
	\text{\rm Depth}(y,o^m,z)\ \ra\ \text{\rm DPivot}(y,o^m)\\
	\text{\rm Depth}(y,o^m,o)\ \ra\ \text{\rm DChange}(y,o^m)\\
	\text{\rm Succ}(x,y,z,o,\bar w,\bar w'),\text{\rm DChange}(y,\bar w'),\text{\rm Depth}(y,\bar w,z)\ \ra\\ \text{\rm DPivot}(y,\bar w)\\
	\text{\rm Succ}(x,y,z,o,\bar w,\bar w'),\text{\rm DChange}(y,\bar w'),\text{\rm Depth}(y,\bar w,o)\ \ra\\ \text{\rm DChange}(y,\bar w)\\
	\text{\rm Succ}(x,y,z,o,\bar w,\bar w'),\text{\rm DPivot}(y,\bar w')\ \ra\ \text{\rm DCopy}(y,\bar w)\\
	\text{\rm Succ}(x,y,z,o,\bar w,\bar w'),\text{\rm DCopy}(y,\bar w')\ \ra\ \text{\rm DCopy}(y,\bar w).
\end{align*}
We can now specify that the depth of non-root nodes is the depth of their parent plus one as follows:
\begin{align*}
	\text{\rm Did}(x,y,z,o,\bar w), \text{\rm NonRoot}(y),\text{\rm DChange}(x,\bar w)\ \ra\\ \text{\rm Depth}(y,\bar w,z)\\
	\text{\rm Did}(x,y,z,o,\bar w), \text{\rm NonRoot}(y),\text{\rm DPivot}(x,\bar w)\ \ra\\ \text{\rm Depth}(y,\bar w,o)\\
	\text{\rm Did}(x,y,z,o,\bar w), \text{\rm NonRoot}(y),\text{\rm DCopy}(x,\bar w),\text{\rm Depth}(x,\bar w,z)\ \ra\\ \text{\rm Depth}(y,\bar w,z)\\
	\text{\rm Did}(x,y,z,o,\bar w), \text{\rm NonRoot}(y),\text{\rm DCopy}(x,\bar w),\text{\rm Depth}(x,\bar w,o)\ \ra\\ \text{\rm Depth}(y,\bar w,o).
\end{align*}

It remains to explain how we start the construction of a new stratum. To this end, we first need to specify the new root nodes, which is done by the following TGD
\begin{align*}
\text{\rm Node}(x,y,z,o),\text{\rm NonMaxStratum}(y)\ \ra\\ \exists w \, \text{\rm Node}(y,w,z,o),\text{\rm NewRoot}(w),
\end{align*}
and, of course, a new root node is a root node
\[
\text{\rm NewRoot}(x)\ \ra\ \text{\rm Root}(x).
\]

We finally need to specify that a new root node belongs to a tree of the next stratum, i.e., we need a mechanism for increasing the stratum counter by one. This can be done similarly to the depth counter. To this end, we need a way to classify digits as pivot, change and copy, which is done via the following TGDs:
\begin{align*}
\text{\rm Node}(x,y,z,o),S_n(y,z)\ \ra\ \text{\rm SPivot}_n(y)\\
\text{\rm Node}(x,y,z,o),S_n(y,o)\ \ra\ \text{\rm SChange}_n(y)
\end{align*}
and, for each $i \in \{2,\ldots,n\}$, we have
\begin{align*}
\text{\rm Node}(x,y,z,o),\text{\rm SChange}_i(y),S_{i-1}(y,z)\ \ra\ \text{\rm SPivot}_{i-1}(y)\\
\text{\rm Node}(x,y,z,o),\text{\rm SChange}_i(y),S_{i-1}(y,o)\ \ra\ \text{\rm SChange}_{i-1}(y)\\
\text{\rm Node}(x,y,z,o),\text{\rm SPivot}_i(y)\ \ra\ \text{\rm SCopy}_{i-1}(y)\\
\text{\rm Node}(x,y,z,o),\text{\rm SCopy}_i(y)\ \ra\ \text{\rm SCopy}_{i-1}(y).
\end{align*}
We can now increment the stratum counter via the following TGDs: for each $i \in \{2,\ldots,n\}$, we have
\begin{align*}
\text{\rm Node}(x,y,z,o),\text{\rm NewRoot}(y),\text{\rm SChange}_i(x)\ \ra\ S_i(y,z)\\
\text{\rm Node}(x,y,z,o),\text{\rm NewRoot}(y),\text{\rm SPivot}_i(x)\ \ra\ S_i(y,o)\\
\text{\rm Node}(x,y,z,o),\text{\rm NewRoot}(y),\text{\rm SCopy}_i(x),S_i(x,z)\ \ra\ S_i(y,z)\\
\text{\rm Node}(x,y,z,o),\text{\rm NewRoot}(y),\text{\rm SCopy}_i(x),S_i(x,o)\ \ra\ S_i(y,o).
\end{align*}
This completes the construction of $\dep_{n,m}$. It can be verified that, for each $\ell,n,m>0$, $\dep_{n,m} \in \class{G} \cap \class{CT}_{D_{\ell}}$. It remains to show the desired lower bound.
To this end, we rely on the following claim, which can be shown via an easy induction. For an atom $\text{\rm Node}(u,v,0,1) \in \chase{D}{\dep_{n,m}}$, where $D = \{\text{\rm Node}(c,c,0,1)\}$, we say that it belongs to stratum $j \in [2^n-1]$ if $\{S_1(v,b_1),\ldots,S_n(v,b_n)\} \subseteq \chase{D}{\dep_{n,m}}$ implies $b_1,\ldots,b_n$ is the binary representation of $j$.

\begin{claim}\label{cla:lower-bound-guarded}
	Let $D = \{\text{\rm Node}(c,c,0,1)\}$. For $j \in \{0,\ldots,2^n-1\}$,
	\begin{eqnarray*}
	&& |\{\bar t \mid \text{\rm Node}(\bar t) \in \chase{D}{\dep_{n,m}} \text{ belongs to stratum } j\}|\\
	&\geq& 2^{\left((j+1) \cdot \left(2^{\left(2^m\right)}-1\right)\right)}.
	\end{eqnarray*}
\end{claim}

By Claim~\ref{cla:lower-bound-guarded}, it is straightforward to see that
\[
|\chase{D_{\ell}}{\dep_{n,m}}|\ \geq\ \ell \cdot 2^{\left(2^n \cdot \left(2^{\left(2^m\right)}-1\right)\right)}
\]
and the claim follows.


\section*{Proof of Theorem~\ref{the:complexity-guarded}}

Consider a database $D$, and a set $\dep \in \class{G}$. By Theorem~\ref{thm:characterization-guarded}, we know that $\dep \in \class{CT}_D$ iff $\gsimple{\dep}$ is $\gsimple{D}$-weakly-acyclic. From the work~\cite{GoMP14}, where the linearization technique has been originally introduced, we know that $\ling{D}$ (respectively, $\ling{\dep}$) can be computed in polynomial time in $|D|$ (respectively, $|\dep|$), exponential time in $|\sch{\dep}|$, and double-exponential time in $\arity{\dep}$. Therefore:
\begin{itemize}
	\item Since the simplification of an atom can be clearly computed in polynomial time, we conclude that $\gsimple{D}$ can be computed in polynomial time in $|D|$, exponential time in $|\sch{\dep}|$, and double-exponential time in $\arity{\dep}$.
	
	\item Since each TGD of $\ling{\dep}$ induces at most $\arity{\dep}^{\arity{\dep}}$ simplifications, we conclude that $\gsimple{\dep}$ can be computed in polynomial time in $|\dep|$, exponential time in $|\sch{\dep}|$, and double-exponential time in $\arity{\dep}$.
\end{itemize}
By Theorem~\ref{the:complexity-sl}, $\mathsf{CT}(\class{SL})$ is in $\textsc{NL} \subseteq \textsc{PTime}$. Thus, we can check whether $\gsimple{\dep}$ is $\gsimple{D}$-weakly-acyclic in polynomial time in $|D|$, exponential time in $|\sch{\dep}|$, and double-exponential time in $\arity{\dep}$, by constructing $\gsimple{D}$ and $\gsimple{\dep}$, and then call the algorithm for $\mathsf{CT}(\class{SL})$.
Hence, $\mathsf{CT}(\class{G})$ is in \textsc{2ExpTime}, in \textsc{ExpTime} for schemas of bounded arity, and in \textsc{PTime} in data complexity.

Recall that the \textsc{2ExpTime} and \textsc{ExpTime} lower bounds are inherited from~\cite{CaGP15}. It remains to explain how we get the \textsc{PTime}-hardness in data complexity. To this end, we recall the problem of propositional atom entailment, parameterized by a class $\class{C}$ of TGDs:\footnote{The input TGDs to this problem can mention $0$-ary predicates.}

\smallskip

\begin{center}
	\fbox{
		\begin{tabular}{ll}
			{\small PROBLEM} : & $\mathsf{PAE}(\class{C})$
			\\
			{\small INPUT} : & A database $D$, a set $\dep \in \class{C}$ of TGDs,\\
			& and a propositional atom $R()$.
			\\
			{\small QUESTION} : &  Is it the case that $R() \in \chase{D}{\dep}$?
	\end{tabular}}
\end{center}

\medskip

From Proposition~15 of~\cite{CaGP15}, which relies on a technique known as the looping operator, we conclude the following: assuming that $\mathsf{PAE}(\class{CT} \cap \class{G})$ is $C$-hard in data complexity (i.e., when $\dep$ and the propositional atom are fixed), where $C$ is a complexity class closed under logspace reductions, then $\mathsf{CT}(\class{G})$ is co$C$-hard in data complexity.
We can inherit from~\cite{CaGL12} that $\mathsf{PAE}(\class{CT} \cap \class{G})$ is \textsc{PTime}-hard in data complexity, and thus, $\mathsf{CT}(\class{G})$ is \textsc{PTime}-hard in data complexity.

\end{document}